%% file: paper.tex
\pdfoutput=1

\documentclass[11pt]{article}

\usepackage{placeins}
\usepackage[numbers,sort&compress]{natbib}
\usepackage{amsthm}
\usepackage{amssymb}
\usepackage{hyperref}
\usepackage{graphicx}
\usepackage{xspace}
\usepackage{url}
\usepackage{subcaption}
\usepackage{appendix}
\usepackage{authblk}
\usepackage{multicol}
\usepackage{listings}
\usepackage{enumitem}

\graphicspath{{./pictures/}}

\usepackage[T1]{fontenc}
\usepackage[tt=false]{libertine}
\usepackage[libertine]{newtxmath}
\usepackage{inconsolata}
\usepackage{geometry}
\usepackage{xcolor}

\usepackage{mathtools}

\newtheorem{theorem}{Theorem}
\newtheorem{lemma}{Lemma}

\usepackage{booktabs}

\title{Crystalline: Fast and Memory Efficient Wait-Free Reclamation}

\author{Ruslan Nikolaev, Binoy Ravindran}
\affil{rnikola@vt.edu, binoy@vt.edu\\Virginia Tech, USA}

\date{}

\begin{document}

\maketitle

\begin{abstract}
\input{abstract.tex}
\end{abstract}

\input{commands.tex}


\definecolor{codegreen}{rgb}{0,0.6,0}
\definecolor{codegray}{rgb}{0.5,0.5,0.5}
\definecolor{codepurple}{rgb}{0.58,0,0.82}
\definecolor{backcolour}{rgb}{0.96,0.96,0.96}

\definecolor{darkpastelred}{rgb}{0.76, 0.23, 0.13}
\definecolor{darkolivegreen}{rgb}{0.33, 0.42, 0.18}

\lstdefinestyle{mystyle}{
frame=l,
framesep=3.2mm,
framexleftmargin=1.3mm,
rulecolor=\color{black},
fillcolor=\color{codegray},
    backgroundcolor=\color{white},   
    commentstyle=\color{darkolivegreen},
    keywordstyle=\color{darkpastelred},
    numberstyle=\tiny\bf\color{white},
    stringstyle=\color{codepurple},
    basicstyle=\fontfamily{LinuxLibertineMonoT-TLF}\fontsize{5.9}{8}\selectfont,
    breakatwhitespace=false,         
    breaklines=true,                 
    captionpos=b,                    
    keepspaces=true,                 
    numbers=left,                    
    numbersep=5pt,                  
    showspaces=false,                
    showstringspaces=false,
    showtabs=false,                  
    tabsize=2
}
 
\lstset{style=mystyle}

\lstset{emph={%  
    done, changed, produced, int64, uint64, thread\_local, invptr, nullptr%
    },emphstyle={\color{darkpastelred}}%
}%

\input{introduction.tex}

\input{background.tex}

\input{blocking.tex}

\input{lreclamation.tex}

\input{wreclamation.tex}

\input{correctness.tex}

\input{evaluation.tex}

\input{related.tex}

\input{conclusion.tex}

\bibliography{lockfree}

\end{document}

%% file: abstract.tex
Historically, memory management based on lock-free reference counting was very inefficient, especially for read-dominated workloads. Thus, approaches such as epoch-based reclamation (EBR), hazard pointers (HP), or a combination thereof have received significant attention. EBR exhibits excellent performance but is blocking due to potentially unbounded memory usage. In contrast, HP are non-blocking and achieve good memory efficiency but are much slower. Moreover, HP are only lock-free in the general case. Recently, several new memory reclamation approaches such as WFE and Hyaline have been proposed. WFE achieves wait-freedom, but is less memory efficient and suffers from suboptimal performance in oversubscribed scenarios; Hyaline achieves higher performance and memory efficiency, but lacks wait-freedom.

We present a new wait-free memory reclamation scheme, Crystalline, that simultaneously addresses the challenges of high performance, high memory efficiency, and wait-freedom. Crystalline guarantees complete wait-freedom even when threads are dynamically recycled, asynchronously reclaims memory in the sense that any thread can reclaim memory retired by any other thread, and ensures (an almost) balanced reclamation workload across all threads. The latter two properties result in Crystalline's high performance and high memory efficiency. Simultaneously ensuring all three properties require overcoming unique challenges which we discuss in the paper.

Crystalline's implementation relies on specialized instructions which are widely available on commodity hardware such as x86-64 or ARM64. Our experimental evaluations show that Crystalline exhibits outstanding scalability and memory efficiency, and achieves superior throughput than typical reclamation schemes such as EBR as the number of threads grows.

%% file: commands.tex
\def\keywords{\vspace{.5em}
{\noindent{\textit{Keywords}:\,\relax%
}}}
\def\endkeywords{\par}

\makeatletter
\newcommand\footnoteref[1]{\protected@xdef\@thefnmark{\ref{#1}}\@footnotemark}
\makeatother

%% file: introduction.tex
\section{Introduction}

Many- and multi-core architectures have become ubiquitous and are used even for the most basic systems today. Exploiting parallelism on such architectures often requires scalable non-blocking data structures as opposed to more traditional, mutual exclusion lock-based designs. Since non-blocking data structures do not use simple mutual exclusion, their memory management becomes highly challenging: a concurrent thread may hold an obsolete pointer to an object which is about to be freed by another thread. Responding to this serious challenge, \emph{safe memory reclamation} (SMR) schemes for unmanaged C/C++ code have been proposed in literature (e.g.,~\cite{Cohen:2018:DSD:3288538.3276513,OrcGC,epoch1,HPPaper,WFE,hyalineFULL,HEPaper,IBRPaper}). However, they typically involve major trade-offs, such as trading off memory efficiency for high throughput (or vice versa).

Exacerbating the problem, only a few schemes~\cite{Cohen:2018:DSD:3288538.3276513,OrcGC,HPPaper,WFE,hyalineFULL,HEPaper} are truly \emph{non-blocking} and with bounded memory usage, i.e., when a suspension of one thread has no adverse impact on progress in other threads. Furthermore,
only one scheme~\cite{WFE} implements \emph{wait-freedom} for a general case.
Wait-freedom, the strongest non-blocking progress guarantee, is fast becoming a non-negotiable property, especially for tail latency-sensitive applications that are increasingly deployed, for example, in cloud and mission-critical settings. Techniques to create fast, wait-free data structures such as Kogan and Petrank's fast-path-slow-path methodology~\cite{Kogan:2012:MCF:2145816.2145835} typically implicitly assume wait-free memory reclamation.

Reclamation schemes with high \textit{performance} and memory \textit{efficiency}\footnote{A lower theoretical memory usage bound does not always imply better practical efficiency, see Section~\ref{sec:evalution}.} are typically preferred. The significance of balancing the reclamation workload -- the task of reclaiming deleted memory objects -- across all threads have not received adequate attention in the literature. Consider the common scenario of read-dominated workloads -- i.e., read operations dominate, but data can still be modified. If the reclamation workload is unbalanced for this case, as in most existing reclamation schemes~\cite{epoch1,HPPaper,HEPaper,IBRPaper}, 
most threads are not actively reclaiming memory, which can cause memory waste. As a side-effect, throughput can also decrease due to the consequent increased pressure on the memory management system.
In oversubscribed scenarios where there are more threads than cores available, this can have a cascading effect: threads that are supposed to reclaim memory are preempted, further degrading performance. 

Vast majority of reclamation schemes (e.g.,~\cite{HPPaper,HEPaper,IBRPaper}) are inherently \emph{synchronous}, i.e., they need to periodically examine which objects marked for deletion can be safely freed. In contrast, reference counting~\cite{refcount4,refcount1,refcount3,refcount2} is
\emph{asynchronous}: a thread with the last reference frees an object. Unfortunately, reference counting is often impractical due to very high overheads when accessing objects. Recent reclamation schemes such as Hyaline~\cite{Hyaline,hyalineFULL} take a somewhat different approach: reference counters are only used when objects are retired. These schemes are still \emph{asynchronous} and exhibit high performance. However, they are still blocking generally speaking (even ``robust'' versions) since memory usage is unbounded if starving threads are not properly managed.

We propose the Crystalline-W scheme, which achieves all of the aforementioned desirable properties of memory reclamation in a single algorithm: wait-freedom, asynchronous reclamation, and balanced reclamation workload. Crystalline-W extends a lock-free algorithm, Crystalline-L, which we also present in the paper (Section~\ref{sec:dc}). Crystalline-L, in turn, is based on Hyaline-1S~\cite{Hyaline,hyalineFULL}, a prior scheme. Making Crystalline wait-free involves overcoming a set of unique challenges, caused in part due to asynchronous memory reclamation (Section~\ref{sec:reclamation}).

We implement Crystalline-L and Crystalline-W and evaluate them on commodity x86-64 hardware (Section~\ref{sec:evalution}).
Both schemes achieve outstanding performance as well as excellent memory efficiency over a broader range of workloads (even outperforming Hyaline schemes). The results are even more outstanding for read-dominated workloads.

%% file: background.tex
\section{Background}
\label{sec:background}

\textbf{Progress.}
Following the traditional \emph{non-blocking} categorization, we
call an algorithm \emph{lock-free} if one or more threads complete an operation after a finite number of steps even in the presence of contending threads, and \emph{wait-free} -- if \emph{all} threads complete their operations after a
finite number of steps. Unsurprisingly, wait-freedom is typically harder to
achieve due to stronger progress properties.
Non-blocking progress should also be considered from a memory usage perspective, which must be bounded~\cite{HPPaper,HEPaper,pedroHEFULL}.
Otherwise, when exhausting memory, no thread can make progress. Such algorithms are \emph{blocking} even if they lack locks.

\textbf{Challenges.}
One of the key challenges of non-blocking memory reclamation
is to guarantee that memory usage is bounded.
EBR~\cite{epoch1} is a widely known scheme which is relatively easy to use, 
but is blocking due to potentially unbounded memory usage. Hazard
pointers (HP)~\cite{HPPaper}, hazard eras (HE)~\cite{pedroHEFULL}, or similar schemes are non-blocking, but they are only
lock-free due to a potentially unbounded loop in a method that safely retrieves pointers. Furthermore, other schemes, such as interval-based reclamation (IBR)~\cite{IBRPaper} and Hyaline-1S~\cite{Hyaline,hyalineFULL}, are not
lock-free in general since underlying data structures need to restart operations to guarantee bounded memory usage for starving threads. Such schemes still protect against threads that are stalled indefinitely and are known
as \emph{robust}.

\begin{table}[htbp]
		\caption{Comparison of Crystalline with existing reclamation schemes.}
\begin{center}
\resizebox{.93\textwidth}{!}{
\begin{tabular}{l l c l c c l}
    \toprule
			\textbf{Scheme} & \textbf{Type} & \textbf{Robust} & \textbf{Progress} & \textbf{Snapshot-Free} & \textbf{Overhead} \\
    \midrule
			EBR & Synchronous & No & Blocking & Yes & 1 word\footnotemark \\
			IBR & Synchronous & Yes & Blocking & No & 3 words \\
			HP & Synchronous & Yes & Lock-Free & No & 1 word\footnoteref{notetab} \\
			HE & Synchronous & Yes & Lock-Free & No & 3 words \\
			WFE & Synchronous & Yes & Wait-Free & No & 3 words \\
			Hyaline-1 & Asynchronous & No & Blocking & Yes & 3 words \\
			Hyaline-1S & Asynchronous & Yes & Blocking & Yes & 3 words \\
	\midrule
			Crystalline-L & Asynchronous & Yes & Lock-Free & Yes & 3 words \\
			Crystalline-W & Asynchronous & Yes & Wait-Free & Yes & 3 words \\
    \bottomrule
\end{tabular}
}
\end{center}
\label{tbl:comparsmr}
\end{table}

\footnotetext{\label{notetab}HP and EBR reserve 1 word for a local limbo
list pointer. This overhead can be eliminated altogether by allocating an
intermediate container object when retiring, but that causes undesirable circular
allocator dependency (avoided in Section~\ref{sec:evalution}). In the same vein, other schemes can reduce the overhead, e.g., IBR, HE, Hyaline-1S, and Crystalline-L's overhead is only 1 word
with container objects. In practice, overheads are often irrelevant due to objects often being cache-line sized to avoid false sharing, causing inevitable memory waste. (Though some exceptions also do exist.)}

\textbf{Snapshot-Freedom.} Many schemes (HP, IBR, HE, WFE) require
taking a \textit{snapshot} of global state (all hazard pointers or eras) to
avoid repeated access of these global variables when checking if
not-yet-freed blocks can be safely deleted. These schemes typically need to
examine fundamentally the \emph{same} state many times per each \emph{iteration}, which results in expensive cache
misses (due to contention) without taking local snapshots.\footnote{Snapshot-free schemes can still examine global state several times but only in different iterations that are sufficiently separated in time, i.e., taking snapshots is not helpful due to a global state divergence. Put differently, this simply means that an efficient implementation is feasible without taking snapshots.}

Taking snapshots results in additional $O(n^2)$ memory usage, which is
substantial as the number of threads, $n$, grows.
EBR, Hyaline-1, and Hyaline-1S do not need snapshots. Note that while IBR, HE, and
WFE (which all extend EBR) trade this quality for robustness,
Hyaline-1S is still snapshot-free, which makes it a great candidate for
further extensions.

\textbf{Comparison.}
Table~\ref{tbl:comparsmr} compares several existing reclamation schemes
against the proposed Crystalline schemes using relevant criteria. Crystalline schemes are \emph{asynchronous}, i.e., a thread that has
the last reference to a memory object ends up freeing it. Since threads that end up freeing memory are more or less arbitrary, the reclamation workload of these schemes is well \emph{balanced} across threads. Balanced reclamation is especially important for read-dominated workloads, wherein most threads are not actively modifying data structures, and enables faster memory reclamation.
Here we implicitly assume that all threads can be treated equally, but this is typical for wait-free algorithms which use helping.

Crystalline-L/-W reserve 3 words in each memory object for
a header. This is analogous to other lock- and wait-free schemes such as HE
and WFE. EBR and HP are more modest but EBR is blocking while HP
is generally much slower.

\textbf{Atomic Instructions.}
\label{sec:atomics}
CAS (compare-and-swap) is used universally by most lock- and wait-free algorithms. FAA (fetch-and-add) and SWAP are two specialized instructions that are often implemented in hardware (e.g., x86-64 and ARM64) and are leveraged by some non-blocking
algorithms~\cite{10.1145/1989493.1989549,Morrison:2013:FCQ:2442516.2442527,WFE,ChaoranWFQ}.
The execution time of hardware-implemented FAA and SWAP is typically bounded, 
allowing them to be used directly in wait-free algorithms.

WCAS (wide CAS)\footnote{Not to be confused with double-word CAS, which updates two \textit{distinct} words and is not widely available.}
is another specialized instruction, which updates two \textit{contiguous}
words, and is available on commodity x86-64 and ARM64 architectures. Typically, it is used
to attach a monotonically increasing \textit{tag} to a pointer to prevent the
ABA problem~\cite{Herlihy:2008:AMP:1734069}. Crystalline-W uses WCAS for a somewhat similar purpose -- tags
identify slow-path cycles and prevent spurious updates.

Note that Crystalline-L does not need specialized instructions.
WFE~\cite{WFE}, the only existent wait-free scheme, requires
WCAS and FAA instructions for wait-free
progress guarantees. Crystalline-W additionally requires the hardware to support SWAP. This is not burdensome in practice as all current architectures that implement WCAS and FAA also implement SWAP. It is reasonable to assume that the same trend will continue for the foreseeable future.

\begin{figure}[ht]
\hspace{4mm}
\begin{subfigure}{.47\columnwidth}
\lstset{language=C++}
\begin{lstlisting}[]
struct StkNode : Node {
	StkNode* next;
	Object*  obj;
};
StkNode* stack_top = nullptr;

void push(Object* obj) {
	StkNode* node = new StkNode;
	node->obj = obj;
	do // Replace the stack top
		node->next = stack_top;
	while (!CAS(&stack_top, node->next, node));
}
\end{lstlisting}
\end{subfigure}
\hspace{2mm}
\begin{subfigure}{.47\columnwidth}
\lstset{language=C++}
\begin{lstlisting}[firstnumber=14]
Object* pop() {
	Object* obj = nullptr;
	while (true) {
		// Fetch the stack top
		StkNode* node = stack_top;
		if (!node) break;
		if (CAS(&stack_top, node, node->next)) {
			obj = node->obj;
			delete node;
			break;
	}	}
	return obj;
}
\end{lstlisting}
\end{subfigure}
\caption{Lock-free stack (without proper reclamation).}
\label{fig:api}
\end{figure}

\textbf{Reclamation Model.}
All discussed schemes expect that the SMR API is called explicitly. Each thread
retains a \emph{reservation}, globally-accessible per-thread state related
to SMR.
Memory blocks, which incorporate all necessary SMR headers, are allocated via standard OS means. After memory blocks appear in a data structure, they can only be reclaimed in two steps. First, after deleting a pointer from the data structure, a memory block is \emph{retired}. When the block finally becomes unreachable by any other concurrent thread, it is returned to the OS.

\subsection{Usage Example}

Throughout our paper we discuss reclamation schemes using a lock-free stack~\cite{Treiber:tech86} shown in Figure~\ref{fig:api}. The stack inserts or deletes
arbitrary objects which are stored in a list consisting of \verb|StkNode|
elements. Each element must incorporate a memory reclamation \verb|Node|, which
is an opaque structure that must be
attached to every reclamation unit.

%% file: blocking.tex
\subsection{Hyaline reclamation schemes}

Four schemes are presented in~\cite{Hyaline,hyalineFULL}. Hyaline
and Hyaline-1 are not robust just like EBR. Two ``robust'' schemes, Hyaline-S and Hyaline-1S, extend two previous algorithms to detect stalled threads with the help of special progress stamps known as \textit{eras}, as discussed below.

The difference between the Hyaline schemes lies in whether threads
can share reservations (Hyaline, Hyaline-S) or need their own per-thread (Hyaline-1, Hyaline-1S) reservations. Hyaline-1 and Hyaline-1S are closer to typical schemes (e.g., EBR, IBR,
HP, and HE) in the way reservations are managed.
In practice, sharing the same reservation is not that useful for
wait-free data structures. Typical wait-free algorithms that
rely on fast-path-slow-path methodologies already maintain per-thread states,
and thus maintaining per-thread reservations do not incur any additional
burden. In this paper, we focus only on Hyaline-1 and Hyaline-1S to increase
the tractability of the problem of wait-free reclamation.

\subsubsection{Hyaline-1}

Hyaline-1's~\cite{Hyaline,hyalineFULL} API is very similar to that of classical EBR~\cite{epoch1} and includes:
\begin{itemize}
\item \verb|activate()|, \verb|clear():| methods used by each thread to enclose
a data structure operation; local pointers to shared memory are only
valid within each such enclosure.
In Figure~\ref{fig:api}, \verb|activate()| must be inserted before \textbf{L8,~L15} and \verb|clear()| after \textbf{L12,~L24}.
\item \verb|retire(blk):| a method which indicates that block
\verb|blk| is deleted and will not be accessed subsequently. Already running threads can still access \verb|blk|; when \verb|blk| can
be safely freed, it will be de-allocated.
In Figure~\ref{fig:api}, \verb|retire()| must be used in lieu of \textbf{L22}.
\end{itemize}

\begin{figure}[ht]
\hspace{2mm}
\begin{minipage}{.48\textwidth}
\lstset{language=C++}
\begin{lstlisting}
struct Node { // Each node takes 3 memory words
	union {
		uint64       refc;       // REFS: Refer. counter
		Node*        bnext;      // SLOT: Next node
	};
	union { // SLOT nodes reuse the space after retire
		uint64       birth;      // Node's birth era
		Reservation* slot;       // SLOT: While retiring
		Node*        next;       // SLOT: After retiring
	};
	Node*  blink;          // REFS: First SLOT node,
};                       // SLOT: Pointer to REFS
struct Batch { // Accumulates nodes before retiring
	Node*  first, refs;  // Init: nullptr, nullptr
	int    counter;      // Init: 0
};
thread_local Batch batch;    // Per-thread batches
\end{lstlisting}
\caption{Hyaline's data structures.}
\label{fig:algnode}
\end{minipage}%
\begin{minipage}{.52\textwidth}
\includegraphics[width=\textwidth]{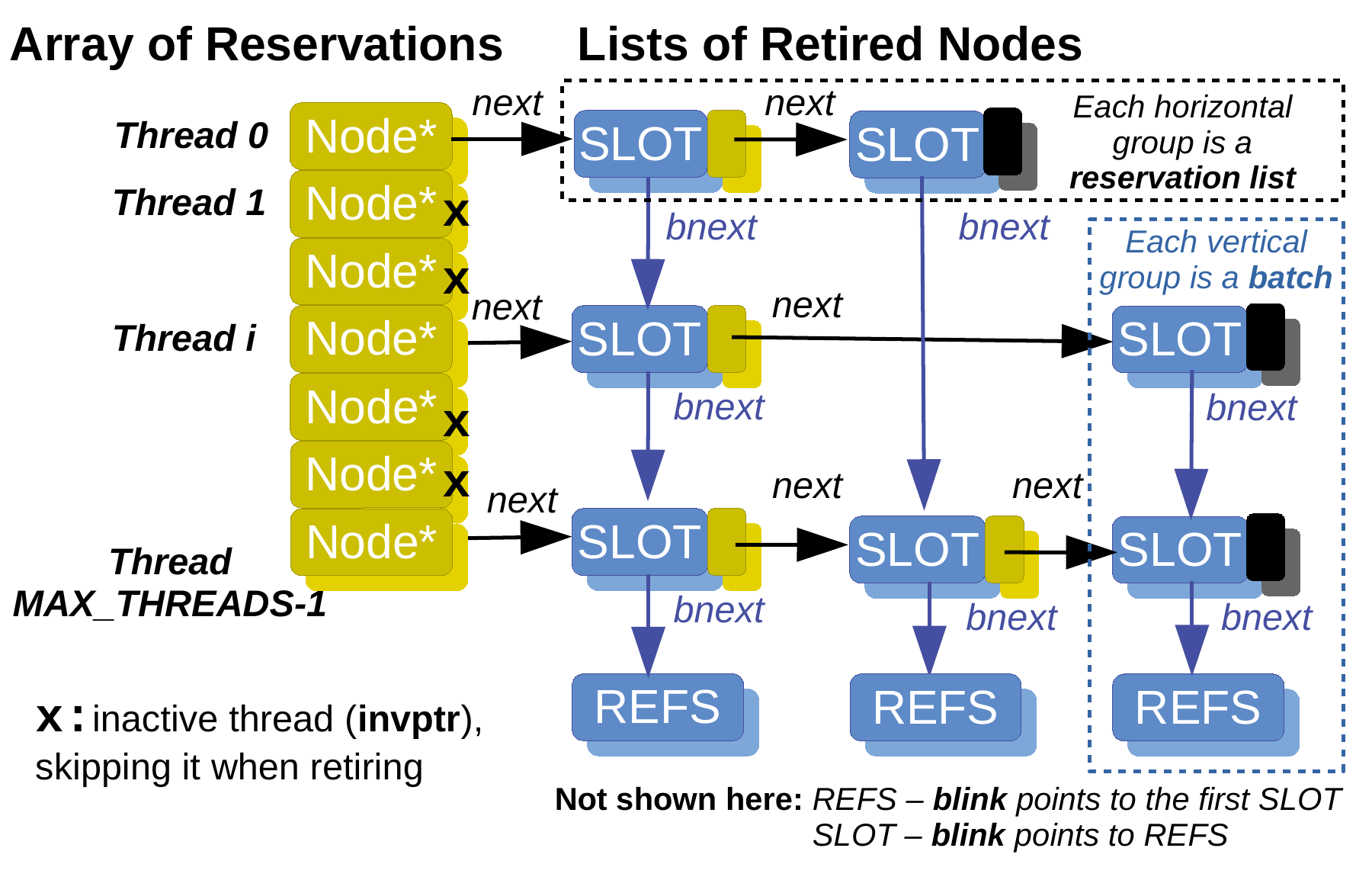}
\caption{An example with retired batches.}
\label{fig:ds}
\end{minipage}%
\end{figure}

Retired memory blocks, \emph{nodes}, are first collected in a thread-local list,
known as a \textit{batch}, before completing their actual retirement.
The purpose of batches is twofold: (i) eliminate contention due to frequent
\verb|retire()| calls by retiring multiple nodes together, and (ii) reserve
a pointer in every node for a separate per-thread linked list when retiring the batch. Due to (ii), the number of nodes in each batch should
at least equal to the number of threads, MAX\_THREADS. Furthermore, an additional node
in the batch is needed to keep a reference counter, which reflects the number
of active threads when the batch is retired. When active threads
complete their operations by calling \verb|clear()|, they traverse their lists
and decrement this reference counter.
Once the reference counter reaches zero, the entire batch is reclaimed.
Thus, batches should at least contain MAX\_THREADS+1 nodes.

Figure~\ref{fig:algnode} shows \verb|Node|'s (header) structure
which is placed in every memory block. \emph{Union}
definitions emphasize how three words are shared for different fields.
For convenience, we
categorize retired nodes into two types: REFS and SLOT. Obviously, no such
distinction exists when nodes are allocated. However, as nodes are retired
and attached to a thread-local batch, we differentiate them.
The very first retired node in the batch is denoted as REFS, and its purpose
is to keep batch's reference counter. All subsequent nodes are denoted as SLOT,
and each of them keeps one list pointer for one reservation. SLOT nodes are linked
together, and each of them has a reference to the REFS node. The REFS node
has a reference to the first SLOT node.
Figure~\ref{fig:ds} shows how retired batches are related
to reservations.

Figure~\ref{fig:algactivate} presents the pseudocode for Hyaline-1. In the code, we use TID to denote the current thread ID.
Each reservation contains its \textit{list} of retired nodes.
The original Hyaline-1's presentation~\cite{Hyaline,hyalineFULL} used a binary reference counter (0 or 1) associated with each reservation's list pointer, which would indicate whether the
corresponding thread is active. Figure~\ref{fig:algactivate} indicates a ``cleared'' reservation by using a special \verb|invptr| value for reservation's list pointer, which is equivalent to that but does not ``steal'' one bit from a pointer. We use the \verb|(void *) -1| value, which is reserved by OSes, e.g., in the \verb|mmap(2)| system call, for errors and will never appear in data structures.
Any legit pointer indicates an ``active'' reservation. Note that \verb|nullptr|
is still legit and simply indicates an empty list.

\begin{figure}[ht]
\hspace{3mm}
\begin{subfigure}{.48\columnwidth}
\lstset{language=C++}
\begin{lstlisting}[]
struct Reservation {
	Node* list; // Init: invptr
};
Reservation rsrv[MAX_THREADS];

void activate() {
	rsrv[TID].list = nullptr; // Active: empty list
}

void clear() {
	Node* p = SWAP(&rsrv[TID].list, invptr);
	traverse(p); // p!=invptr, activate() was called
}

void traverse(Node* next) {
	while (next != nullptr) {
		Node* curr = next;
		next = curr->next;
		Node* refs = curr->blink;
		if (FAA(&refs->refc, -1) == 1) free_batch(refs);
}	}

void free_batch(Node* refs) {
	// RNODE is dummy for Hyaline-1, Hyaline-1S,
	Node* n = RNODE(refs->blink); // and Crystalline-L
	do {
		Node* obj = n;
		// refc and bnext overlap and are 0
		// (nullptr) for the last REFS node
		n = n->bnext;
		free(obj);
}	} while (n != nullptr);
\end{lstlisting}
\end{subfigure}
\hspace{2mm}
\begin{subfigure}{.48\columnwidth}
\lstset{language=C++}
\begin{lstlisting}[firstnumber=33]
// This addend protects from premature reclamation
const uint64 REFC_PROTECT = 1 << 63;

void retire(Node* node) {
	if (!batch.first) { // the REFS node
		batch.refs = node;
		node->refc = REFC_PROTECT;
	} else { // SLOT nodes
		node->blink = batch.refs; // points to REFS
		node->bnext = batch.first;
	}
	batch.first = node;
	// Must have MAX_THREADS+1 nodes to insert to
	// MAX_THREADS lists, exit if do not have enough
	if (batch.counter++ < MAX_THREADS) return;
	// blink of REFS points to the 1st SLOT node
	batch.refs->blink = batch.first;
	Node* curr = batch.first;
	int64 cnt = -REFC_PROTECT;
	for (int i = 0; i < MAX_THREADS; i++) {
		while (true) {
			Node* prev = rsrv[i].list;
			if (prev == invptr) break;
			curr->next = prev;
			if (CAS(&rsrv[i].list, prev, curr)) {
				cnt++;
				break;
		}	}
		curr = curr->bnext;
	}
	if (FAA(&batch.refs->refc, cnt) == -cnt) // Finish
		free_batch(batch.refs); // retiring: change refc
	batch.first = nullptr; batch.counter = 0;
}
\end{lstlisting}
\end{subfigure}
\vspace{-10pt}
\caption{The Hyaline-1 reclamation scheme.}
\label{fig:algactivate}
\label{fig:algretire}
\end{figure}

\subsubsection{Hyaline-1S}
If one thread stalls while using Hyaline-1, new batches will continue to be
added to the list that corresponds to the stalled thread.
Thus, the reference counter of every added batch will account for
that stalled list, preventing its reclamation.
HE~\cite{HEPaper} and IBR~\cite{IBRPaper} proposed to use
\textit{birth} and \textit{retire} eras that determine intervals for
memory block lifetimes. Hyaline-1S also adopted \textit{birth}
eras (but without \textit{retire} eras) to guard against
stalled threads.

Hyaline-1S's~\cite{Hyaline,hyalineFULL} API is very similar to that of IBR and extends Hyaline-1 with:
\begin{itemize}
\item \verb|protect(ptr):| method used to safely retrieve pointers that
are about to be dereferenced by creating a local copy. In Figure~\ref{fig:api},
we need to wrap \verb|stack_top| with this method in \textbf{L18} when assigning it to \verb|node|. 
\item \verb|alloc_node():| a method which wraps \verb|malloc()| to initialize
an object with its \emph{birth} era. In Figure~\ref{fig:api},
it must be called in lieu of \textbf{L8}.
\end{itemize}

All \emph{valid} pointers (i.e., not marked anyhow for deletion) retrieved between \verb|activate()| and \verb|clear()| can be safely accessed. This semantics, common in other robust schemes such
as hazard pointers or IBR, is different from that of Hyaline-1 and EBR.
Care must be taken when accessing pointers, specifically when traversing
logically ``deleted'' nodes as in Harris' linked-list~\cite{HarrisList}, which
needs to be modified to promptly unlink nodes as discussed in~\cite{HPPaper}.

Figure~\ref{fig:algrob} shows Hyaline-1S's changes to Hyaline-1's corresponding methods.
Each reservation additionally adds a 64-bit active \textit{era} that was last observed by a respective thread.
The eras are assumed to never overflow in practice.
When nodes are allocated, \verb|alloc_node()| initializes their birth eras
with the global era clock value. When retrieving pointers, threads call
\verb|protect()| to update reservation's era value.

\verb|retire()| calculates batch's minimum birth era,
which is used subsequently when the entire batch is retired.
Unlike in IBR or HE, the \textit{birth} era field does not need to survive the
\verb|retire()| call.
REFS reuses this field to keep the minimum birth era, while SLOT reuses the space to keep a reservation list pointer.

\begin{figure}[ht]
\hspace{3mm}
\begin{subfigure}{.45\columnwidth}
\lstset{language=C++}
\begin{lstlisting}[]
struct Reservation { // Add the 'era' field
	Node* list; // Init: invptr
	uint64 era; // Init: 0
};

void retire(Node* node) {
	if (!batch.first) { // the REFS node
		...
	} else { // SLOT nodes
		// Reuse the birth era of REFS to retain
		// the minimum birth era in the batch
		if (batch.refs->birth > node->birth)
			batch.refs->birth = node->birth;
		...
	}
	uint64 min_birth = batch.refs->birth;
	...
	// After Line 55, Figure 4, add:
	if (rsrv[i].era < min_birth) break;
}
\end{lstlisting}
\end{subfigure}
\hspace{2mm}
\begin{subfigure}{.51\columnwidth}
\lstset{language=C++}
\begin{lstlisting}[firstnumber=21]
uint64 global_era = 1;
thread_local int alloc_cnt = 0;

// Replace malloc() with alloc_node()
Node* alloc_node(int size) {
	if (!(alloc_cnt++ % ALLOC_FREQ)) FAA(&global_era, 1);
	Node* node = malloc(size);
	node->birth = global_era;
	return node;
}

Node* protect(Node** obj) {
	uint64 prev = rsrv[TID].era;
	while (true) {
		Node* node = *obj;
		uint64 era = global_era;
		if (prev == era) return node;
		rsrv[TID].era = era;
		prev = era;
}	}
\end{lstlisting}
\end{subfigure}
\caption{The Hyaline-1S reclamation scheme (showing changes only with respect to Hyaline-1).}
\label{fig:algrob}
\end{figure}

%% file: lreclamation.tex
\section{Lock-Free Reclamation with Crystalline-L}

\label{sec:dc}

Hyaline-1S is still not lock-free generally speaking unless operations are
restarted for starving threads. The problem
arises when a thread performs operations that make more and more reservations in an unbounded loop (e.g., consider one ``unlucky'' thread that is stuck traversing a list because it keeps growing). Though restarting is not unreasonable
and certain data structures can be modified accordingly~\cite{PEBR,hyalineFULL,IBRPaper}, other
schemes such as HP and HE are lock-free without restarting.
Comparing HP/HE with IBR/Hyaline-1S, we can see that
the former two prevent unbounded memory usage due to slight API differences, which result in finer granularity of reservations. Would it be possible to
adopt HP/HE's API for Hyaline-1S? Surprisingly, not only we can do it (but we need to overcome challenges discussed below), such an algorithm is \emph{faster} and more memory \emph{efficient} than Hyaline-1S, as shown in Section~\ref{sec:evalution}.

We redefine the API such that each
thread bounds the number of retrieved local pointers. This API, which is similar to that of HP, has the following subtle differences with Hyaline-1S:
\begin{itemize}
\item \verb|protect()|: additionally passes an \emph{index} that is assigned to a specific
local pointer. \verb|protect()| no longer has a cumulative effect, each time
it resets any previous reservation associated with the specified index.
\item \verb|activate()|: \textbf{removed} from the API completely, as \verb|protect()| is now
non-cumulative.
\item \verb|clear()|: resets \emph{all} (rather than just one) local pointer reservations made by \verb|protect()|.
\end{itemize}

Note that Figure~\ref{fig:api} defines only one local pointer in \textbf{L18}, for which index 0 is used.
With this API, each thread holds at most MAX\_IDX local pointers (indices are 0..MAX\_IDX-1).
For each local pointer,
Crystalline-L maintains a separate list and era.
Due to the stricter API, even starving threads do not reserve
unbounded memory, making the scheme lock-free (Section~\ref{sec:correctness}).
Figure~\ref{fig:algdc} shows Crystalline-L's changes to Hyaline-1S. When retiring, a batch must be attached to \emph{all} reservation indices. Therefore, before it can be retired, a batch would need to have MAX\_THREADS$\times$MAX\_IDX+1 rather than MAX\_THREADS+1 nodes.

\begin{figure*}[ht]
\hspace{3mm}
\begin{subfigure}{.48\columnwidth}
\lstset{language=C++}
\begin{lstlisting}[]
Reservation rsrv[MAX_THREADS][MAX_IDX];

void clear() {
	for (int i = 0; i < MAX_IDX; i++) {
		Node* p = SWAP(&rsrv[TID][i].list, invptr);
		if (p != invptr) traverse(p);
}	}

void try_retire() { // Attempt to retire a batch
	uint64 min_birth = batch.refs->birth;
	Node* last = batch.first;
	// Check if we have a sufficient number of nodes
	for (int i = 0; i < MAX_THREADS; i++) {
		for (int j = 0; j < MAX_IDX; j++) {
			if (rsrv[i][j].list == invptr) continue;
			if (rsrv[i][j].era < min_birth) continue;
			if (last == batch.refs)
				return; // Ran out of nodes, exit
			last->slot = &rsrv[i][j];
			last = last->bnext;
	}	}
	// Retire if successful
	Node* curr = batch.first;
	int64 cnt = -REFC_PROTECT;
	for (; curr != last; curr = curr->bnext) {
		Reservation* slot = curr->slot;
		while (true) { // Do not check 'era' again,
			Node* prev = slot->list; // it can only grow
			if (prev == invptr) break;
			curr->next = prev;
			if (CAS(&slot->list, prev, curr)) {
				cnt++;
				break;
	}	}	}
	if (FAA(&batch.refs->refc, cnt) == -cnt) // Finish
		free_batch(batch.refs); // retiring: change refc
	batch.first = nullptr; batch.counter = 0;
}

\end{lstlisting}
\end{subfigure}
\hspace{2mm}
\begin{subfigure}{.495\textwidth}
\lstset{language=C++}
\begin{lstlisting}[firstnumber=39]
Node* protect(Node** obj, int index) {
	uint64 prev_era = rsrv[TID][index].era;
	while (true) {
		Node* ptr = *obj;
		uint64 curr_era = global_era;
		if (prev_era == curr_era) return ptr;
		prev_era = update_era(curr_era, index);
}	}

// Clean up the old list and set a new era
uint64 update_era(uint64 curr_era, int index) {
	if (rsrv[TID][index].list != nullptr) {
		Node* list = SWAP(&rsrv[TID][index].list, nullptr)
		if (list != invptr) traverse(list);
		curr_era = global_era;
	}
	rsrv[TID][index].era = curr_era;
	return curr_era;
}

void retire(Node* node) {
	if (!batch.first) { // the REFS node
		batch.refs = node;
		node->refc = REFC_PROTECT;
	} else { // SLOT nodes
		// Reuse the birth era of REFS to retain
		// the minimum birth era in the batch
		if (batch.refs->birth > node->birth)
			batch.refs->birth = node->birth;
		node->blink = batch.refs; // points to REFS
		node->bnext = batch.first;
	}
	batch.first = node;
	if (batch.counter++ % RETIRE_FREQ == 0) {
		// blink of REFS points to the 1st SLOT node
		batch.refs->blink = RNODE(batch.first);
		try_retire();
}	}
\end{lstlisting}
\end{subfigure}
\caption{The Crystalline-L reclamation scheme (showing changes only with respect to Hyaline-1S).}
\label{fig:algdc}
\end{figure*}

One problem with Hyaline-1S is that a batch must aggregate \emph{all} MAX\_THREADS+1 nodes before retirement can even start. If this number is too high, the algorithm may become suboptimal since retirement is delayed. This is further aggravated with Crystalline-L.

In practice, the required number of nodes is much
lower as each node is appended to the respective list only if list's era overlaps with batch's minimum birth era. In Crystalline-L, retirement can be tried sooner, irrespective of the number of threads. It implements \textit{dynamic} batches to avoid their excessive growth.
On average, batch sizes roughly equal the number of cores as eras for preempted threads are going to be behind. Other schemes (HP/IBR/HE/WFE)
amortize their scanning frequency on a similar scale for good performance.

\verb|retire()| periodically calls \verb|try_retire()|, which checks
how many reservation lists\footnote{We call lists of retired objects ``reservation lists'' since there are multiple lists in Crystalline-L/-W, one list per a reservation.} are to be changed for the batch to be retired and
records the location of each such reservation slot. If the number of nodes
in the batch suffices, \verb|try_retire()| completes the retirement
by appending the nodes to their corresponding slot lists.
\verb|try_retire()| may look similar to the method
in EBR, IBR, HE, or HP which periodically peruses thread-local lists to check
if any of the previously retired nodes are safe to delete. However, it is used for \textit{retirement} rather than \textit{de-allocation} in Crystalline-L.
Furthermore, \verb|try_retire()| has the previously mentioned upper bound
on nodes and is \emph{guaranteed} to eventually succeed after a \textit{finite} number of retries.

%% file: wreclamation.tex
\section{Wait-Free Reclamation with Crystalline-W}

\label{sec:reclamation}

Our new wait-free reclamation scheme, Crystalline-W, extends Crystalline-L described in Section~\ref{sec:dc}.

\textbf{Assumptions.}
Similar to~\cite{WFE}, we assume a 64-bit CPU that supports
WCAS (to manipulate 64-bit eras) and wait-free FAA. Additionally, the CPU must support wait-free SWAP as previously discussed in Section~\ref{sec:atomics}.
These requirements are fully met by commonplace x86-64 and ARM64 architectures. Architectures without these instructions can always fall back to Crystalline-L, which of course forfeits wait-freedom.

Similar to WFE~\cite{WFE}, Crystalline-W slightly alters \verb|protect()|'s API. We pass an additional parameter, \textit{parent}, which refers
to a parent object where the hazardous reference is located (\verb|nullptr| is allowed for the static topmost locations).

\textbf{Challenges.}
The Crystalline-L scheme is only lock-free because of the two fundamental challenges, which we solve in the next two sections:
\begin{itemize}
    \item \verb|retire()| calls \verb|try_retire()|, which has an unbounded loop (Lines~25-34) when retiring nodes. Contention occurs if two \verb|update_era()| or \verb|try_retire()| modify the same list.
    \item \verb|protect()| has an unbounded loop (Lines~41-45) which must converge on the era value. \verb|alloc_node()| unconditionally changes the era clock to bound memory usage.
\end{itemize}

\subsection{List Tainting}

\verb|try_retire()|, shown in Figure~\ref{fig:algdc}, requires a CAS loop (Lines~25-34), which contends with unconditional SWAP in \verb|update_era()| (Line~51) or concurrent \verb|try_retire()|. Our solution to this problem is to use unconditional SWAP in \verb|try_retire()| as shown in Figure~\ref{fig:tainting}. We initialize the \textit{next} field of a retired node with \verb|nullptr|, and swap the current list head with the pointer to the retired node.

Aside from a corner case with the (\verb|invptr|) cleared reservation (for which we have to roll back unless Line~51 already takes it, see Section~\ref{sec:pseudo} for more details), two major cases need to be considered, as shown in Figure~\ref{fig:tainting}. If the retired node still has its next field intact, it simply attaches the previous list as its tail (Figure~\ref{fig:tainting}, part I). It is also possible that a thread associated with the reservation already called \verb|traverse()| for the just retired node. Crystalline-W's \verb|traverse()| additionally taints the retrieved \textit{next} pointer of all traversed nodes by using SWAP. (Note that a node is \textit{merely} traversed and dereferenced by the thread; only when the batch's reference counter reaches 0 is when the node is deallocated.) \verb|try_retire()| still holds the batch with the retired node, and it finds that the next field is tainted (Figure~\ref{fig:tainting}, part II). Thus, \verb|try_retire()| traverses the tail on behalf of the thread that just finished \verb|traverse()|.

\begin{figure*}[ht]
\begin{subfigure}{.49\textwidth}
\includegraphics[width=.9\columnwidth]{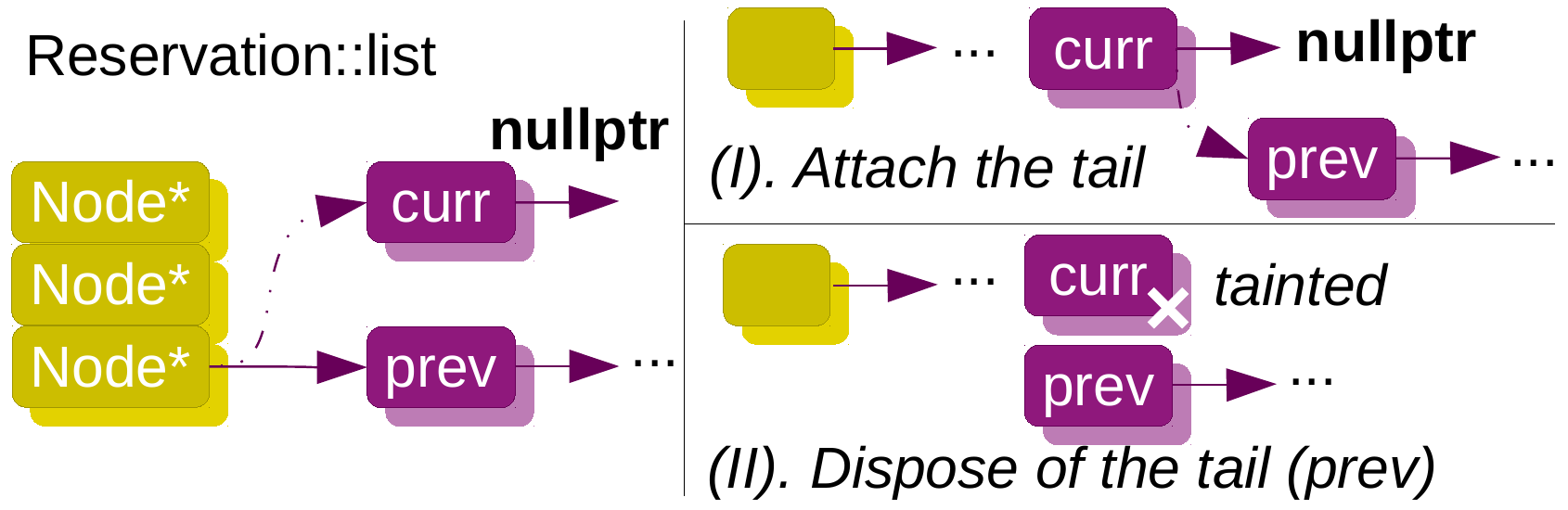}
\caption{List tainting}
\label{fig:tainting}
\end{subfigure}
\begin{subfigure}{.49\textwidth}
\includegraphics[width=\columnwidth]{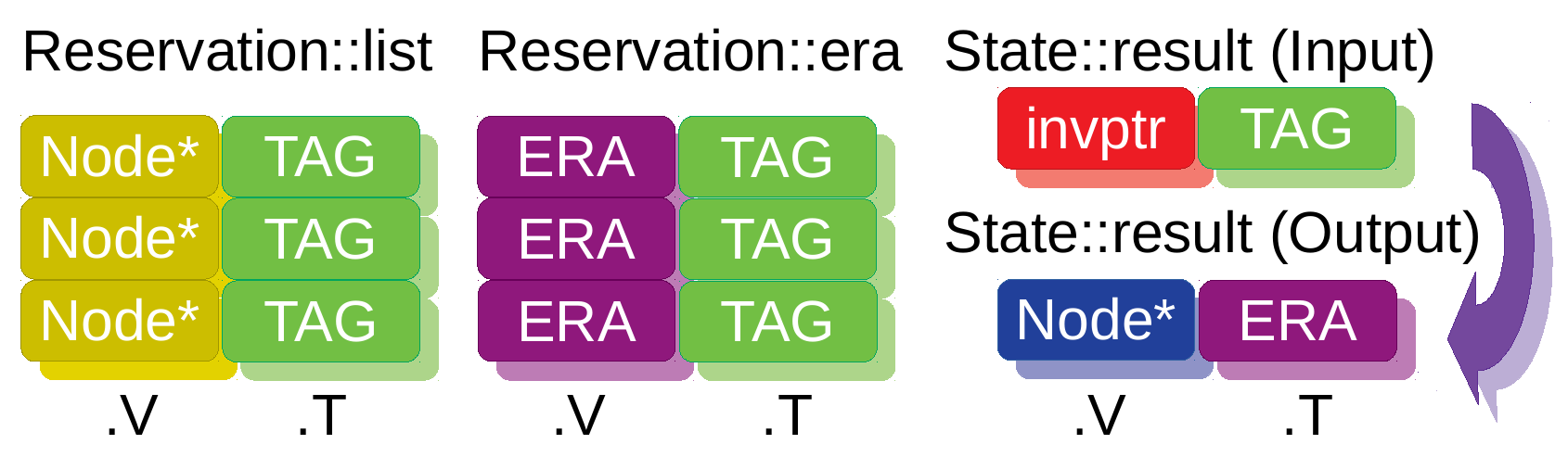}
\caption{Reservation and state format}
\label{fig:state}
\end{subfigure}
\caption{Crystalline-W algorithm.}
\end{figure*}

\subsection{Slow Path}

The second problem of making \verb|protect()| wait-free was previously considered by WFE~\cite{WFE} for hazard eras. The approach, which we also adopt in Crystalline-W, is based on a fast-path-slow-path idea, wherein \verb|alloc_node()|, instead of unconditionally incrementing the global era, runs a helper method \verb|increment_era()|. \verb|increment_era()| calls \verb|help_thread()| for every thread that needs helping, and only after that it increments the global era.

For a finite number of iterations, \verb|protect()| attempts to converge
on the era clock. After that, a slow path is taken by calling \verb|slow_path()|. Then, \verb|help_thread()| collaborates with \verb|slow_path()| to help converge \verb|protect()|.
Crystalline-W's approach, however, has major differences with that of WFE which we discuss below. Section~\ref{sec:pseudo} presents pseudocode.

\textbf{Data Formats.}
The \emph{rsrv} array is modified to contain pairs \{.V, .T\} by attaching \emph{tags} to both the list and era fields as shown in Figure~\ref{fig:state}. These tags identify the slow path cycle and are used to prevent spurious updates as in WFE.
For the Crystalline-W methods that are only shown in Figure~\ref{fig:algdc}, 
the .V component must be implied.
The \emph{rsrv} array is also extended by two special reservations, i.e., each thread has $MAX\_IDX+2$ indices. The two extra reservations are used internally by \verb|help_thread()|.

As in WFE, each thread maintains its \emph{state} for the slow path.
The \emph{result} field of \emph{state}  is used for both input and output. On input,
a current slow path cycle is advertised. Output contains a retrieved pointer with the corresponding era. In Figure~\ref{fig:state}, we demonstrate how input and output values must be aligned. To distinguish the two cases unambiguously, \verb|invptr| is placed as a pointer on input.
Finally, \emph{slow\_counter} counts the number of threads in the slow path. This is used to optimize \verb|increment_era()|.

\textbf{Retirement Status.}
\verb|alloc_node()| initializes \verb|blink|. When  its value is changed in \verb|retire()|, it will no longer be \verb|nullptr|. \verb|blink| will point to REFS (for SLOT nodes). This indicates that the node
is already in the process of retirement. Furthermore, the reference counter in REFS becomes immediately reachable from SLOT. To identify the REFS node itself, we steal one bit from the \verb|blink| field. (See new RNODE definitions for Crystalline-W.)

\textbf{Object Hand-Over.}
Crystalline-W uses one special reservation for the parent object and the other one for a retrieved object. However, a race condition happens if the parent object is retired by the time the corresponding reservation is made in \verb|help_thread()|. Likewise, a race condition happens in the opposite direction when assigning the retrieved object to the thread which is in \verb|slow_path()|.
Note that WFE avoids it by simply scanning the list of retired nodes twice, which is impossible with Crystalline-W.

These race conditions, obviously, are only considered for retired objects. To properly hand-over a parent object from \verb|slow_path()|, we maintain the \emph{parents} array. During the slow path initialization, a pointer to the parent is passed through \emph{state}.
When \verb|helper_thread()| enters, it will initialize its entry in \emph{parents} with this value. Subsequently, when \verb|slow_path()| exits, \verb|handover_parent()| is called. The latter iterates through \emph{parents} and tries to replace a matching parent to \verb|nullptr| with CAS. If successful, the reference counter of the parent is incremented. (REFC\_PROTECT\_HANDOVER protects from a premature object de-allocation.)
The other side subsequently detects the hand-over and dereferences  the object.

In a similar problem, when \verb|help_thread()| needs to pass the retrieved pointer back to \verb|slow_path()|, we access the actual retrieved (but already retired) node and increment its batch reference counter. Since we know the exact node, and \verb|protect()| is not yet complete, we clean up the existing reservation list and attach this batch's REFS as a ``terminal node.''

\textbf{REFS-Terminal Nodes.}
The idea behind REFS-terminal nodes is that they can only appear at the very end of the list. REFS-terminal nodes can be in as many lists as desirable. We steal one bit from the preceding pointer to indicate a REFS-terminal node. When encountering this node, \verb|traverse()| immediately terminates.

\input{algorithm.tex}

%% file: algorithm.tex
\subsection{Implementation}

\label{sec:pseudo}

Crystalline-W's high-level changes (with respect to Crystalline-L) are shown
in Figure~\ref{fig:algwfdcapi}.
Note that reservation's list and era are now tagged. Tags are used in
slow-path procedures only; fast-path procedures simply use the value component.
Crystalline-W defines per-thread \emph{state} (for each corresponding
reservation) used in slow-path procedures and \emph{slow\_counter} to
identify if any thread needs helping. Those are somewhat similar to
WFE's~\cite{WFE} corresponding slow-path variables. Finally, Crystalline-W
defines the \emph{parent} array to facilitate object handover, as discussed
previously. Object handover is unique to Crystalline-W since it cannot simply
scan the list of retired objects twice, as WFE, to avoid race conditions.
Figure~\ref{fig:algwfdcapi} also modifies \verb|alloc_node()| to internally
call \verb|increment_era()| in lieu of doing FAA on the global era
directly. Finally, \verb|protect()| calls \verb|slow_path()| if it fails
to converge after MAX\_TRIES.

\begin{figure*}[ht]
\hspace{3mm}
\begin{subfigure}{.48\columnwidth}
\lstset{language=C++}
\begin{lstlisting}[]
template <typename type> struct Tag {
	type   V;   // Value
	uint64 T;   // Tag, or Era for State::result
};

struct Reservation {
	Tag<Node*>  list;  // Init: {.V = nullptr, .T = 0}
	Tag<uint64> era;   // Init: {.V = 0, .T = 0}
};

struct State {
	Tag<void*> result; // Init: {.V = nullptr, .T = 0}
	uint64     era;    // Init: 0
	Node*      parent; // Init: nullptr
	Node**     obj;    // Init: nullptr
};

Reservation rsrv[MAX_THREADS][MAX_IDX+2];
State state[MAX_THREADS][MAX_IDX];
Node* parents[MAX_THREADS]; // Init: (all) nullptr
int slow_counter = 0;
\end{lstlisting}
\end{subfigure}
\hspace{2mm}
\begin{subfigure}{.487\columnwidth}
\lstset{language=C++}
\begin{lstlisting}[firstnumber=22]
// Help other threads before incrementing the era
Node* alloc_node(int size) {
	if (alloc_cnt++ % ALLOC_FREQ == 0) increment_era();
	Node* node = malloc(size);
	node->birth = global_era;
	node->blink = nullptr; // Retired if != nullptr
	return node;
}

// Use the fast-path-slow-path method
Node* protect(Node** obj, int index, Node* parent) {
	int tries = MAX_TRIES;
	uint64 prev_era = rsrv[TID][index].era.V;
	while (--tries != 0) {
		Node* ptr = *obj;
		uint64 curr_era = global_era;
		if (prev_era == curr_era) return ptr;
		prev_era = update_era(curr_era, index);
	}
	return slow_path(obj, index, parent);
}
\end{lstlisting}
\end{subfigure}
\caption{Crystalline-W (API function changes).}
\label{fig:algwfdcapi}
\end{figure*}

Figure~\ref{fig:algwfdretire} shows changes to \verb|try_retire()| and \verb|traverse()|. These methods use list tainting, as previously discussed.
Also, unlike WFE~\cite{WFE}, we use two slow-path tag transitions (odd and even). This is needed to make the number of iterations finite in some loops (Lemmas~\ref{theorem:prop3} and \ref{theorem:prop4}) by collaborating with \verb|try_retire()| which will skip odd tags.
Figure~\ref{fig:algwfdretire} also shows
\verb|increment_era()|'s implementation as used by \verb|alloc_node()|.

\begin{figure*}[ht]
\hspace{3mm}
\begin{subfigure}{.495\textwidth}
\lstset{language=C++}
\begin{lstlisting}[]
// Redefine RNODE to encode REFS links: steal one bit
// to indicate REFS nodes (also applies to other
// functions that previously used dummy RNODE)
#define IS_RNODE(x) (x & 0x1) // Check if a REFS link
#define RNODE(x)    (x ^ 0x1) // Encode or decode REFS

// Another huge addend for the slow path
// (in addition to previously defined REFC_PROTECT)
const uint64 REFC_PROTECT_HANDOVER = 1 << 62;

// Adds a special REFS-terminal node and list tainting
void traverse(Node* next) {
	while (next != nullptr) {
		Node* curr = next;
		if (IS_RNODE(curr)) { // REFS-terminal node
			// It is always the last node, exit
			Node* refs = RNODE(curr);
			if (FAA(&refs->refc, -1) == 1) free_batch(refs);
			break;
		}
		next = SWAP(&curr->next, invptr); // Tainting
		Node* refs = curr->blink;
		if (FAA(&refs->refc, -1) == 1) free_batch(refs);
}	}

// Increments the global era, replaces regular FAA
// (needs to help other threads first)
void increment_era() {
	if (slow_counter != 0) {
		for (int i = 0; i < MAX_THREADS; i++) {
			for (int j = 0; j < MAX_IDX; j++) {
				if (state[i][j].result.V == invptr)
					help_thread(i, j);
	}	}	}
	FAA(&global_era, 1);
}
\end{lstlisting}
\end{subfigure}
\hspace{2mm}
\begin{subfigure}{.478\textwidth}
\lstset{language=C++}
\begin{lstlisting}[firstnumber=37]
void try_retire() { // This replacement is wait-free
	uint64 min_birth = batch.refs->birth;
	Node* last = batch.first;
	// Also check odd tags to bound slow-path loops
	for (int i = 0; i < MAX_THREADS; i++) {
		for (int j = 0; j < MAX_IDX+2; j++) {
			if (rsrv[i][j].list.V == invptr ||
			    (rsrv[i][j].list.T & 0x1)) continue;
			if (rsrv[i][j].era.V < min_birth ||
			    (rsrv[i][j].era.T & 0x1)) continue;
			if (last == batch.refs)
				return; // Ran out of nodes, exit
			last->slot = &rsrv[i][j];
			last = last->bnext;
	}	}
	// Retire, make it wait-free by list tainting
	Node* curr = batch.first;
	int64 cnt = -REFC_PROTECT;
	for (; curr != last; curr = curr->bnext) {
		Reservation* slot = curr->slot;
		if (slot->list.V == invptr) continue;
		Node* prev = SWAP(&slot->list.V, curr);
		if (prev != nullptr) {
			if (prev == invptr) { // Inactive previously
				if (CAS(&slot->list.V, curr, invptr))
					continue; // Try to rollback
			} else { // Tainted: traverse the chopped tail
				if (!CAS(&curr->next, nullptr, prev))
					traverse(prev);
		}	}
		cnt++;
	}
	if (FAA(&batch.refs->refc, cnt) == -cnt)
		free_batch(batch.refs);
	batch.first = nullptr; batch.counter = 0;
}
\end{lstlisting}
\end{subfigure}
\vspace{-5pt}
\caption{Crystalline-W's try\_retire(), traverse(), and increment\_era().}
\label{fig:algwfdretire}
\end{figure*}

\begin{figure*}[ht]
\hspace{3mm}
\begin{subfigure}{.48\textwidth}
\lstset{language=C++}
\begin{lstlisting}[]
// Hand over the parent object if it is retired
void handover_parent(Node* parent) {
	if (parent && parent->blink != nullptr) {
		Node* refs = get_refs_node(parent);
		FAA(&refs->refc, REFC_PROTECT_HANDOVER);
		int64 cnt = -REFC_PROTECT_HANDOVER;
		for (int i = 0; i < MAX_THREADS; i++)
			if (CAS(&parents[i], parent, nullptr)) cnt++;
		FAA(&refs->refc, cnt);
}	}

uint64 get_birth_era(Node* node) { // Get parent's
	if (node == nullptr) return 0;   // birth era
	uint64 birth_era = node->birth;
	Node* link = node->blink;
	// For already retired SLOT nodes, use REFS' value
	if (link != nullptr && !IS_RNODE(link))
		birth_era = link->birth;
	return birth_era;
}
\end{lstlisting}
\end{subfigure}
\hspace{2mm}
\begin{subfigure}{.495\textwidth}
\lstset{language=C++}
\begin{lstlisting}[firstnumber=21]
// Makes the tag+1 transition and detaches an old list
void detach_nodes(int i, int j, int tag) {
	// A simple era tag transition: tag -> tag+1
	CAS(&rsrv[i][j].era.T, tag, tag+1);
	// Detach nodes and increment the list tag
	do { // Bounded by MAX_THREADS (try_retire checks
		old = rsrv[i][j].list; // the era tag)
		if (old.T != tag) break;
		bool success = WCAS(&rsrv[i][j].list,
		                     old, { nullptr, tag+1 });
	} while (!success);
	return success ? old.V : invptr; // Previous value
}

// Get REFS node from any node in a batch
Node* get_refs_node(Node* node) {
	Node* refs = node->blink;
	if (IS_RNODE(refs)) refs = node; // This node itself
	return refs;
}
\end{lstlisting}
\end{subfigure}
\vspace{-5pt}
\caption{Crystalline-W's utility functions for the slow path.}
\label{fig:algwfdc1}
\end{figure*}

\begin{figure*}[ht]
\hspace{2.5mm}
\begin{subfigure}{.487\textwidth}
\lstset{language=C++}
\begin{lstlisting}[]
void slow_path(Node** obj, int index, Node* parent) {
	// Getting parent's birth is tricky: for non-retir-
	// ed nodes use 'birth', else retrieve the minimum
	// birth from REFS, see get_birth_era() for details
	uint64 parent_birth = get_birth_era(parent);
	FAA(&slow_counter, 1);
	state[TID][index].obj = obj;
	state[TID][index].parent = parent;
	state[TID][index].era = parent_birth;
	uint64 tag = rsrv[TID][index].era.T;
	state[TID][index].result = { invptr, tag };
	uint64 prev_era = rsrv[TID][index].era.V;
	do { // Bounded by MAX_THREADS
		Node* list, * ptr = *obj;
		uint64 curr_era = global_era;
		if (curr_era == prev_era &&
		    WCAS(&state[TID][index].result,
		          { invptr, tag }, { nullptr, 0 })) {
			rsrv[TID][index].era.T = tag+2;
			rsrv[TID][index].list.T = tag+2;
			FAA(&slow_counter, -1);
			return ptr; // DONE
		}
		// Dereference previous nodes and update the era
		if (rsrv[TID][index].list.V != nullptr) {
			list = SWAP(&rsrv[TID][index].list.V,
			             nullptr);
			if (rsrv[TID][index].list.T != tag)
				goto produced; // Result was just produced
			if (list != invptr) traverse(list);
			curr_era = global_era;
		}
		// WCAS fails only when the result is produced
		WCAS(&rsrv[TID][index].era,
		      { prev_era, tag }, { curr_era, tag });
		prev_era = curr_era;
	} while (state[TID][index].result.V == invptr);
	list = detach_nodes(TID, index, tag); //tag+1 state
produced:
	ptr = state[TID][index].result.V;
	uint64 era = state[TID][index].result.T;
	rsrv[TID][index].era.V = era;
	rsrv[TID][index].era.T = tag+2;
	rsrv[TID][index].list.T = tag+2;
	// Check if the obtained node is already retired
	if (ptr && ptr->blink != nullptr) {
		Node* refs = get_refs_node(ptr);
		FAA(&refs->refc, 1);
		if (list != invptr) traverse(list);
		list = SWAP(&rsrv[TID][index].list.V,
		     RNODE(refs)); // Put a REFS-terminal node
	}
	FAA(&slow_counter, -1);
	// Traverse the previously detached list
	if (list != invptr) traverse(list);
	// Hand over the parent to all helper threads
	handover_parent(parent);
	return ptr; // DONE
}
\end{lstlisting}
\end{subfigure}
\hspace{2mm}
\begin{subfigure}{.51\textwidth}
\lstset{language=C++}
\begin{lstlisting}[firstnumber=60]
void help_thread(int i, int j) {
	Tag<void*> result = state[i][j].result;
	if (result.V != invptr) return;
	uint64 era = state[i][j].era;
	Node* parent = state[i][j].parent;
	if (parent != nullptr) {
		rsrv[TID][MAX_IDX].list.V = nullptr;
		rsrv[TID][MAX_IDX].era.V = era;
		parents[TID] = parent; // Advertise for a handover
	}
	Node** obj = state[i][j].obj;
	uint64 tag = rsrv[i][j].era.T;
	if (tag != result.T) goto changed;
	uint64 curr_era = global_era;
	do { // Bounded by MAX_THREADS
		prev_era = update_era(curr_era, MAX_IDX+1);
		Node* ptr = *obj;
		uint64 curr_era = global_era;
		if (prev_era == curr_era) {
			if (WCAS(&state[i][j].result, // Published the
			          result, { ptr, curr_era })) { // result
				Node* list = detach_nodes(i, j, tag);
				if (list != invptr) traverse(list);
				do { // Set the new era, <= 2 iterations
					old = rsrv[TID][index].era;
					if (old.T != tag+1) break;
				} while (!WCAS(&rsrv[TID][index].era,
				                old, { curr_era, tag+2}));
				// If the obtained node is already retired
				if (ptr && ptr->blink != nullptr) {
					Node* refs = get_refs_node(ptr);
					FAA(&refs->refc, 1);
					do { // Bounded by MAX_THREADS
						old = rsrv[TID][index].list;
						if (old.T != tag+1) break;
						ok = WCAS(&rsrv[TID][index].list,
						           old, { RNODE(refs), tag+2 }));
						if (ok && old.V != invptr) traverse(old.V);
						if (ok) goto done;
					} while (!ok);
					FAA(&refs->refc, -1); // Already inserted
				} else { // A simple tag transition
					CAS(&rsrv[TID][index].list.V, tag+1, tag+2);
			}	}
			break;
		}
	} while (state[i][j].result == result);
done:
	Node* lst = SWAP(&rsrv[TID][MAX_IDX+1].list.V,invptr)
	traverse(lst);
changed: // If handover occurs, dereference the parent
	if (parent != nullptr) {
		if (SWAP(&parents[TID], nullptr) != parent) {
			Node* refs = get_refs_node(parent);
			if (FAA(&refs->refc, -1) == 1) free_batch(refs);
		}
		Node* lst = SWAP(&rsrv[TID][MAX_IDX].list.V,invptr)
		traverse(lst);
}	}
\end{lstlisting}
\end{subfigure}
\caption{Crystalline-W's slow-path methods.}
\label{fig:algwfdc2}
\end{figure*}

Crystalline-W's slow-path and helper thread routines are shown
in Figure~\ref{fig:algwfdc2}. These routines use several utility
methods shown in Figure~\ref{fig:algwfdc1}.
The general idea is similar to that of WFE~\cite{WFE}, with one
major difference: we use the \emph{parent} array to keep parent
references to facilitate object hand-overs, as previously discussed.

Utility methods in Figure~\ref{fig:algwfdc1} are needed to facilitate
the slow path. \verb|get_birth_era()| uses a trick to retrieve
the birth era irrespective of whether the node is retired.
WFE~\cite{WFE} always keeps the birth era. However, the Hyaline and Crystalline
schemes recycle the birth era field after retirement so that they still
use 3 words per each memory object.
Since the birth era does not survive node retirements (except REFS which
stores the minimum era for the batch), that presents a challenge for
Crystalline-W which needs to transfer the parent's era in \verb|slow_path()|, and the parent object can already end up being retired.
We use the following trick. When the parent node is still \emph{not} retired,
we simply retrieve the birth era from the node. Otherwise, we retrieve REFS'
value of the minimum era.

%% file: correctness.tex
\section{Correctness}

\label{sec:correctness}

Both Crystalline-L and Crystalline-W are based on Hyaline-1S, for which correctness is discussed in~\cite{hyalineFULL}. In this section, we present
memory bounds arguments which are critical for lock-free progress
guarantees. We further discuss aspects that pertain to wait-free progress.

\begin{theorem}
Crystalline-L and Crystalline-W are fully memory bounded.\footnote{This is not generally true for Hyaline-1S due to starving threads.}
\end{theorem}

\begin{proof}
Batch sizes are constrained
by MAX\_THREADS$\times$MAX\_IDX+1 nodes. In the worst case, each
batch is
attached to every reservation. The total number of reservations is
MAX\_THREADS$\times$MAX\_IDX. Consequently, the memory usage is bounded by
(MAX\_THREADS$\times$MAX\_IDX+1)$^2$.
Since era increments are also amortized, the total cost is RETIRE\_FREQ$\times$(MAX\_THREADS$\times$MAX\_IDX+1)$^2$. (Note that for Crystalline-W, MAX\_IDX+2 should be used rather than MAX\_IDX since the total number of indices is bigger.)

Note that this theoretical upper bound is worse than that of HE~\cite{pedroHEFULL} or WFE~\cite{WFE}. However, in practice, batches do not accumulate this (worst-case) number of nodes and are retired much faster, often resulting in better practical efficiency. Regardless of that, this worst-case bound is still reasonable and finite.
\end{proof}

\begin{lemma}
\verb|traverse()| calls are wait-free bounded.
\label{lemma:traverse}
\end{lemma}

\begin{proof}
The loop in \verb|traverse()| is bounded by the length of the list of retired batches at a given reservation. \verb|try_retire()| attaches (Line~58, Figure~\ref{fig:algwfdretire}) only those batches
for which the minimum birth era overlaps with the reservation's era.
(Note that every time the era is updated by \verb|update_era()|, the list is emptied.)
Since the eras are periodically incremented in \verb|alloc_node()|, the number of such nodes, and consequently -- batches, is finite.
\end{proof}

\begin{lemma}
\verb|detach_nodes()|'s loop is bounded by MAX\_THREADS iterations.
\label{theorem:prop3}
\end{lemma}

\begin{proof}
\verb|detach_nodes()| makes the tag->tag+1 (odd) transition in the slow path.
The CAS operation in Line~24, Figure~\ref{fig:algwfdc1}, changes the era tag
		unless it was already changed by a concurrent thread. For the loop in Lines~26-31, Figure~\ref{fig:algwfdc1}
to continue, the list tag should have not yet moved to the tag+1 state. Because
		the era tag moves to the tag+2 (even) state only after that transition (Line~43 or Lines~83-87, Figure~\ref{fig:algwfdc2}), i.e., after \verb|detach_nodes()| in Line~38~or~81 (Figure~\ref{fig:algwfdc2}), the era tag
is still tag+1. Consequently, all contending threads in \verb|try_retire()| will
		find that the era tag is odd (Line~46, Figure~\ref{fig:algwfdretire}) and will skip the corresponding reservation
for retirement. (Note that skipping nodes is safe because this race window is
		handled later by Lines~46~and~89, Figure~\ref{fig:algwfdc2}.) Only threads that are already in-progress will proceed and potentially contend because of unconditional list pointer updates in Line~58, Figure~\ref{fig:algwfdretire}. The number of such threads is bounded by MAX\_THREADS, as subsequent \verb|try_retire()| calls will observe that the era tag is odd.
\end{proof}

\begin{lemma}
The loop in Lines~92-99, Figure~\ref{fig:algwfdc2} is bounded by at most
MAX\_THREADS iterations.
\label{theorem:prop4}
\end{lemma}

\begin{proof}
The proof is similar to that of Lemma~\ref{theorem:prop3}. The only difference is that it makes the tag+1->tag+2 (even) transition in the slow path.
Consequently, \verb|try_retire()| will find that the list tag is odd (Line~44, Figure~\ref{fig:algwfdretire}).
\end{proof}

\begin{lemma}
The loop in Lines~13-37, Figure~\ref{fig:algwfdc2} is bounded by at most
MAX\_THREADS iterations.
\label{theorem:prop1}
\end{lemma}

\begin{proof}
In Line~11, Figure~\ref{fig:algwfdc2}, a thread advertises that it needs help. The loop in Lines~13-37 can only fail to converge because of \emph{global\_era} updates. At most MAX\_THREADS
already in-progress threads
		are executing \emph{increment\_era()} from \emph{alloc\_block()}, prior to Line~35, Figure~\ref{fig:algwfdretire}, which updates \emph{global\_era}, but after Line~32, Figure~\ref{fig:algwfdretire}, which detects
what threads need helping. All these threads will execute Line~35.
		That will cause the loop in Lines~13-37 (Figure~\ref{fig:algwfdc2}) fail and repeat.
However, all newer \emph{increment\_era()} calls will only update \emph{global\_era} after \emph{help\_thread()} is complete.
\end{proof}

\begin{lemma}
The loop in Lines~74-106, Figure~\ref{fig:algwfdc2} is bounded by at most
MAX\_THREADS iterations.
\label{theorem:prop2}
\end{lemma}

\begin{proof}
The same idea as in Lemma~\ref{theorem:prop1}. We also need to make sure that the loop will not go beyond one slow path cycle. This is achieved by comparing the tag component in Line~106.
\end{proof}

\begin{theorem}
\verb|retire()| is wait-free bounded.
\end{theorem}

\begin{proof}
\verb|retire()| periodically calls \verb|try_retire()|.
The loop in Lines~55-67, Figure~\ref{fig:algwfdretire}, is bounded by the number of nodes in a batch. (MAX\_THREADS$\times$MAX\_IDX+1 at most). 
Extra RETIRE\_FREQ-1 nodes can be retired since
\verb|retire()| calls \verb|try_retire()| with
the corresponding frequency.
Regardless of the status of the CAS operations (Lines~61~and~64), the loop
moves on to the next node. The \verb|traverse()| call (Line~65) is bounded due to Lemma~\ref{lemma:traverse}.
\end{proof}

\FloatBarrier

\begin{theorem}
\verb|alloc_node()| is wait-free bounded.
\end{theorem}

\begin{proof}
\verb|alloc_node()| calls \verb|increment_era()|. The latter includes a bounded loop with calls to \verb|help_thread()|. Finally,
\verb|help_thread()| is bounded due to Lemmas~\ref{theorem:prop3},~\ref{theorem:prop4},~\ref{theorem:prop2}.
\end{proof}

\begin{theorem}
\verb|protect()| is wait-free bounded.
\end{theorem}

\begin{proof}
The fast path includes a finite number of iterations. It may call \verb|update_era()|, which calls \verb|traverse()|.  \verb|traverse()| is bounded due to Lemma~\ref{lemma:traverse}. \verb|slow_path()| contains a loop which is bounded due to Lemma~\ref{theorem:prop1}. \verb|slow_path()| also calls \verb|detach_nodes()|, which is bounded due to Lemma~\ref{theorem:prop3}. Finally, \verb|slow_path()| can call \verb|traverse()|, which is bounded due to Lemma~\ref{lemma:traverse}.
\end{proof}

\begin{theorem}
\verb|clear()| is wait-free bounded.
\end{theorem}

\begin{proof}
The method has a bounded loop which can call
 \verb|traverse()|, which is bounded due to Lemma~\ref{lemma:traverse}.
\end{proof}

%% file: evaluation.tex
\section{Evaluation}

\label{sec:evalution}

We evaluated the Crystalline schemes from 1 to 192 threads on a
96-core machine consisting of four Intel Xeon E7-8890~v4 2.20~GHz
CPUs, 256GB of RAM. (HyperThreading is OFF for more reliable measurements.)
Threads are pinned in order, socket by socket. We
used clang 9.0.1 with the \verb|-O3| optimizations. (We also ran tests with gcc 9.2.1, but found that clang performs marginally better for
\textit{all} reclamation schemes due to its better code optimization for
C++11 atomics.) Similar to~\cite{WFE,IBRPaper}, we
used jemalloc~\cite{jemalloc} due to its better performance.

We implemented Crystalline-L/-W reclamation schemes in C++11 and integrated them into the benchmark in~\cite{IBRPaper}. We incorporated  additional tests from~\cite{WFE}. We measured the performance of our schemes and compared them against existing well-known or closely related schemes (Figure~\ref{fig:legendsmr}). We skipped: (i) OS-based approaches since they are inevitably blocking; (ii) PEBR~\cite{PEBR} and lock-free garbage collectors due to significant API differences and lack of performance benefits~\cite{hyalineFULL} compared to Hyaline, and consequently -- Crystalline.

The original benchmark code we used~\cite{WFE,IBRPaper} was suboptimal for HP, HE, and WFE due to excessive cache misses when scanning lists of retired nodes. As in~\cite{hyalineFULL}, we take per-thread snapshots of hazard
pointers (or eras) before scanning the list. This improves performance greatly,
especially for HP. A similar optimization already existed for IBR. Finally, EBR, Hyaline-1/-1S, and Crystalline-L/-W are snapshot-free, i.e., do not need this optimization.

Data structures implement abstract key-value interfaces such as
\verb|insert()|-and-\verb|delete()|, \verb|push()|-and-\verb|pop()|, or $ $
\verb|get()|-and-\verb|put()|.
For each data point, the benchmark in~\cite{IBRPaper}
prefills the data structure with 50,000 elements and runs 10 seconds.
Each thread then randomly chooses the corresponding abstract operation.
The key for each operation is randomly chosen from the
range (0, 100,000). We run the experiment 5 times and report the average.
For both Crystalline-W and WFE, the fast path threshold (MAX\_TRIES) is 16.

\begin{figure}
\begin{minipage}{.5\textwidth}
\resizebox{\textwidth}{!}{
		\setlength\tabcolsep{1.5pt}
  \begin{tabular}{ll} \hline
		  \textbf{None} & no reclamation (leaking memory) \\
			\textbf{Hyaline-1} & (non-robust) Hyaline-1~\cite{Hyaline,hyalineFULL} \\
			\textbf{Hyaline-1S} & (robust) Hyaline-1S~\cite{Hyaline,hyalineFULL} \\
			\textbf{HP} & the hazard pointers scheme~\cite{HPPaper} \\
			\textbf{HE} & the hazard eras scheme~\cite{HEPaper} \\
			\textbf{IBR} & 2GEIBR (interval-based)~\cite{IBRPaper} \\
			\textbf{WFE} & the wait-free eras scheme~\cite{WFE} \\
			\textbf{EBR} & epoch-based reclamation\\ \hline
  \end{tabular}
}
  \caption{Evaluated reclamation schemes.}
   \label{fig:legendsmr}
\end{minipage}%
\begin{minipage}{.5\textwidth}
\begin{subfigure}{.5\textwidth}
\includegraphics[width=.99\textwidth]{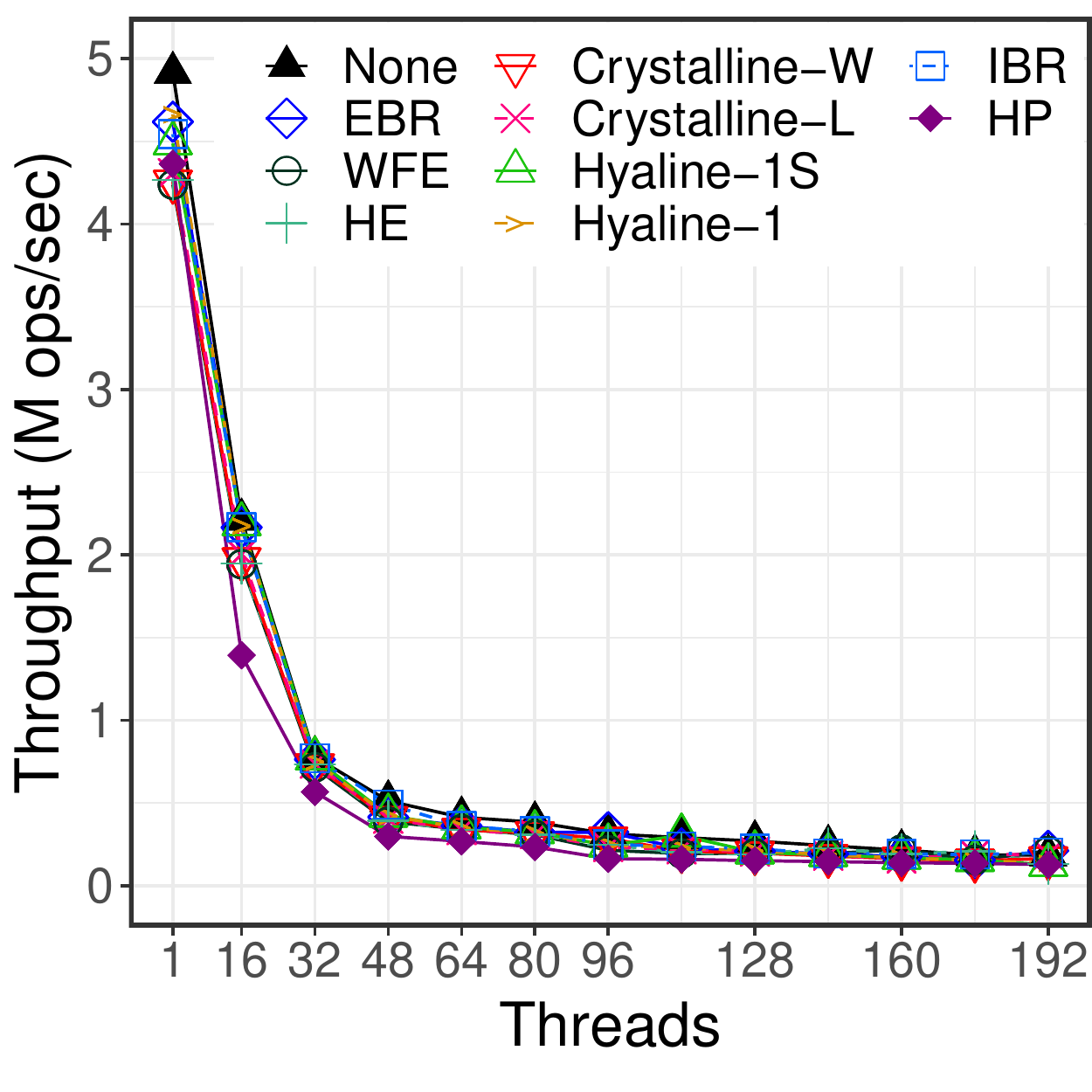}
\vspace{-15pt}
\caption{Throughput}
\label{fig:crturn_throughput}
\end{subfigure}%
\begin{subfigure}{.5\textwidth}
\includegraphics[width=.99\textwidth]{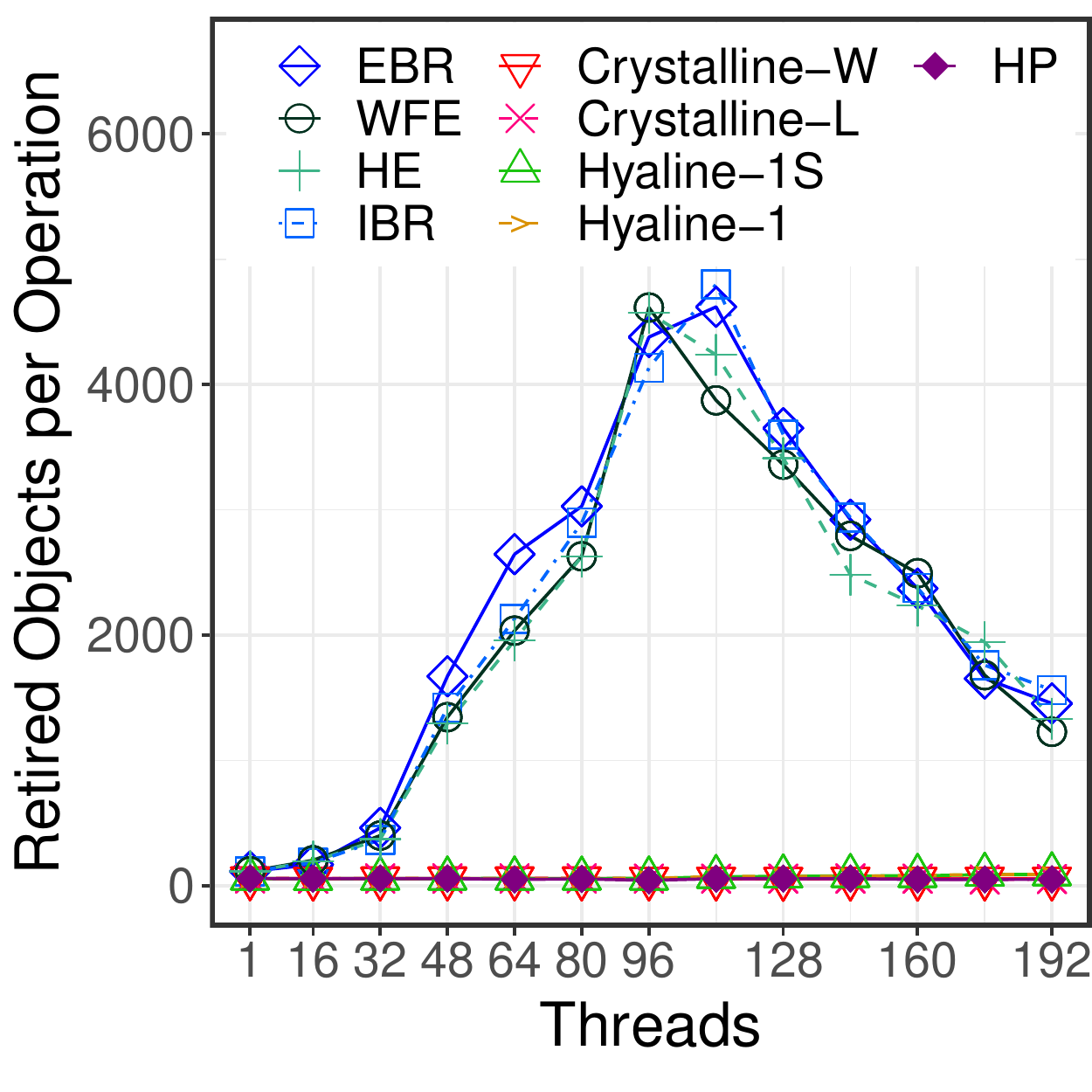}
\vspace{-15pt}
\caption{Retired objs}
\label{fig:crturn_retired}
\end{subfigure}%
\vspace{-5pt}
\caption{Wait-free CRTurnQueue.}
\label{fig:crturn}
\end{minipage}%
\end{figure}

Benchmark parameters must be properly tuned for a fair comparison.
We found that the default parameters used in~\cite{WFE,IBRPaper} are not
optimal for our 96-core machine even for existing schemes, so we adjusted them.
We found that \verb|epochf|=110 and \verb|emptyf|=120, as used by the benchmark, appear to be
optimal for \emph{all} existing schemes because they attain the best possible throughput with good memory efficiency.
The same parameters are also optimal for all Crystalline schemes. Therefore, all schemes were tested with identical parameters. (Note that \verb|emptyf| is used as
RETIRE\_FREQ for \verb|try_retire()| in Crystalline-L/-W.)

We also slightly modified the way memory objects are
retired in the existing schemes. The original benchmark used an extra indirection
when retiring by allocating thread-local list nodes (they store pointers to retired objects), which creates circular allocator dependency.
We did not find this change to make any significant impact on the results for the existing schemes. However, we made this change for a fair comparison with Hyaline and Crystalline.

In the evaluation, we focus on common performance metrics including  throughput and memory reclamation efficiency.
Wait-free data structures are typically more difficult to implement and involve additional variables in their evaluation. Thus, we primarily focus on existing lock-free
data structures from the benchmark in~\cite{IBRPaper}: 
a hash map~\cite{HPPaper}, sorted list~\cite{HarrisList,HPPaper}, and Natarajan binary search tree~\cite{NatarajanTree}. 
However, we also evaluate Ramalhete and Correia~\cite{pedroWFQUEUE}'s CRTurnQueue, which was also previously evaluated with WFE~\cite{WFE}. CRTurnQueue is a wait-free queue which has bounded memory usage, unlike other queues such as WFQUEUE~\cite{ChaoranWFQ}, which has unbounded memory usage~\cite{pedroWFQUEUEFULL}.
CRTurnQueue is faster~\cite{WFE,pedroWFQUEUE} than other wait-free queues with bounded memory usage~\cite{kpWFQUEUE}.

We first ran a write-dominated workload (50\% of \verb|insert()| and 50\% of \verb|delete()| operations). This
workload substantially stresses memory reclamation algorithms. We found that Crystalline-L/-W generally outperform other schemes in both throughput and memory efficiency. Hyaline-1/-1S are often on par but can be worse than Crystalline-L/-W due to
larger granularity and lack of dynamic batches present in more advanced Crystalline schemes.
Though Crystalline's theoretical memory bound is worse than that of HE/WFE (Section~\ref{sec:correctness}), its practical efficiency is better than that of HE/WFE and is often comparable to HP's.

Retired objects are computed per an \emph{individual} node in Hyaline/Crystalline. The number of retired objects is \emph{averaged}
per an \emph{operation} and can even be lower than the batch size.

Figure~\ref{fig:crturn} presents CRTurnQueue results. The write-intensive workload is typical for a queue and guarantees that the queue will not grow indefinitely. Queues generally
do not scale that well, and throughput is almost identical (Figure~\ref{fig:crturn_throughput}) for all algorithms. Crystalline and Hyaline schemes show exceptional memory efficiency, which is on par with HP (Figure~\ref{fig:crturn_retired}).

For the sorted list, hazard pointers has the worst throughput (Figure~\ref{fig:list_throughput}), and WFE exhibits a slight
overhead which is discussed in~\cite{WFE}. All other schemes, except HP, achieve the maximum throughput.
All Hyaline and Crystalline schemes are very memory efficient (Figure~\ref{fig:list_retired}).
For the hash map, Crystalline-L/-W achieve superior throughput (Figure~\ref{fig:hashmap_throughput}), which is especially evident for oversubscribed scenarios, where the
gap with other algorithms is as large as 2x. HP's throughput is worse than that of Crystalline-L/-W. Hyaline-1/-1S's throughput is worse than that of
Crystalline, which can be explained by a smaller granularity of reservations. WFE has the worst throughput. Memory efficiency (Figure~\ref{fig:hashmap_retired}) of Crystalline-L/-W is superior to all algorithms except HP. Hyaline-1/-1S's efficiency is also great before oversubscription but greatly reduces afterwards.
A similar trend is observed for Natarajan tree (Figures~\ref{fig:natarajan_throughput}~and~\ref{fig:natarajan_retired}).
Crystalline-L/-W and Hyaline-1/-1S are more memory efficient before hitting the oversubscribed situation. In the oversubscribed situation, however, the throughput also remains higher than that of other algorithms which implies that the total number of allocated objects was higher in the first place.
As the throughput reduces to that of EBR (192 threads), the number of unreclaimed objects also evens out.

\begin{figure*}[ht]
\begin{subfigure}{.25\textwidth}
\includegraphics[width=.99\textwidth]{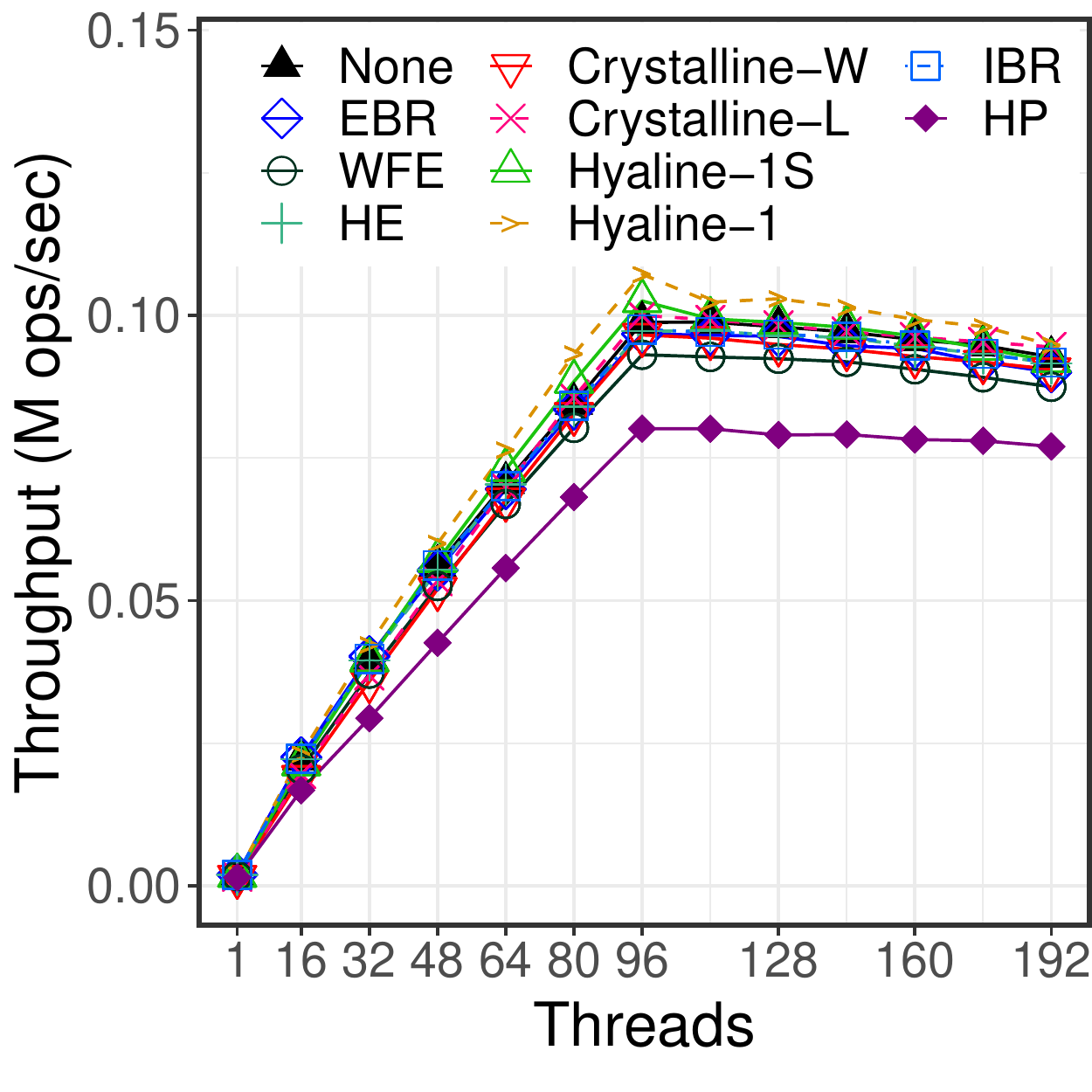}
\vspace{-15pt}
\caption{Throughput (write)}
\label{fig:list_throughput}
\end{subfigure}%
\begin{subfigure}{.25\textwidth}
\includegraphics[width=.99\textwidth]{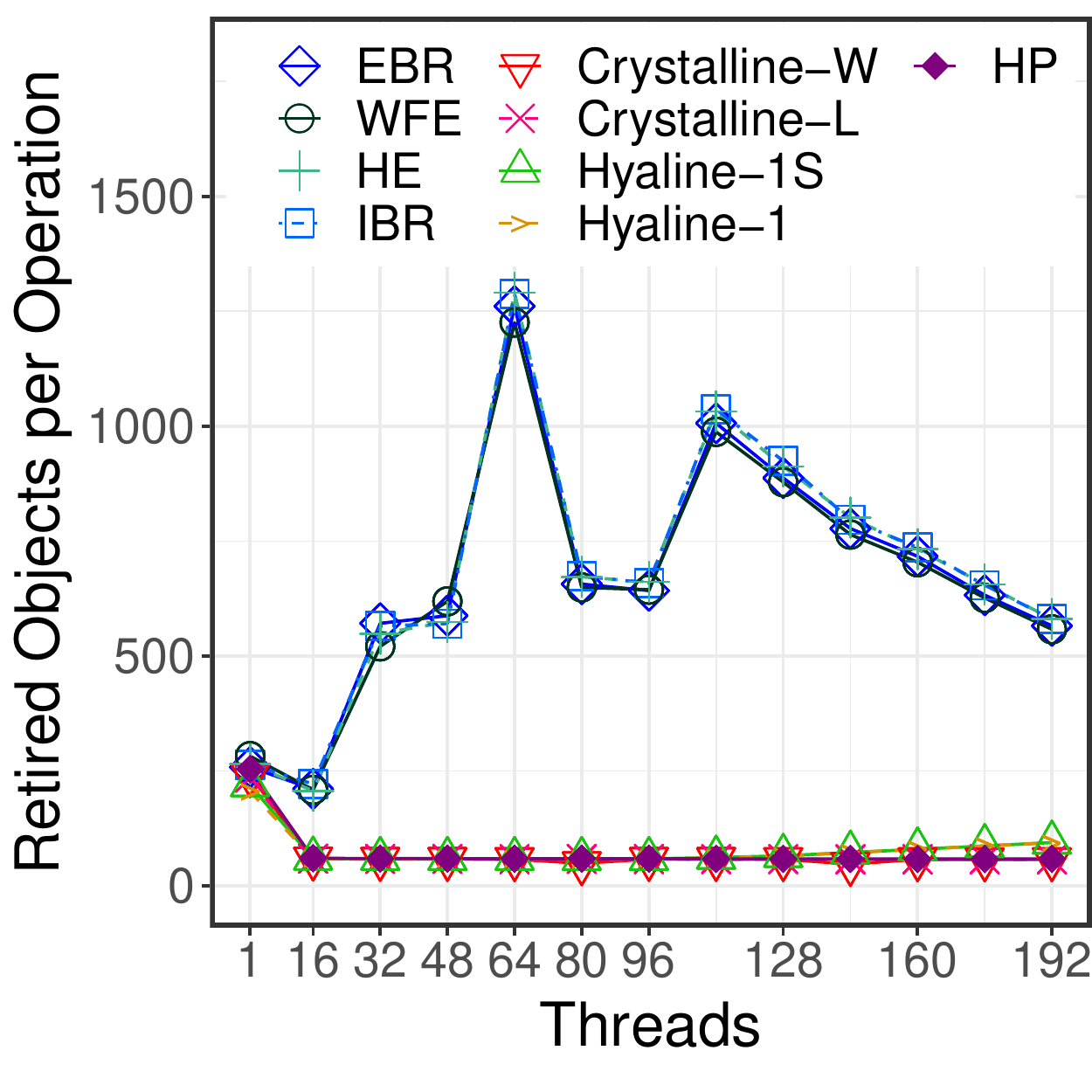}
\vspace{-15pt}
\caption{Retired objs (write)}
\label{fig:list_retired}
\end{subfigure}%
\begin{subfigure}{.25\textwidth}
\includegraphics[width=.99\textwidth]{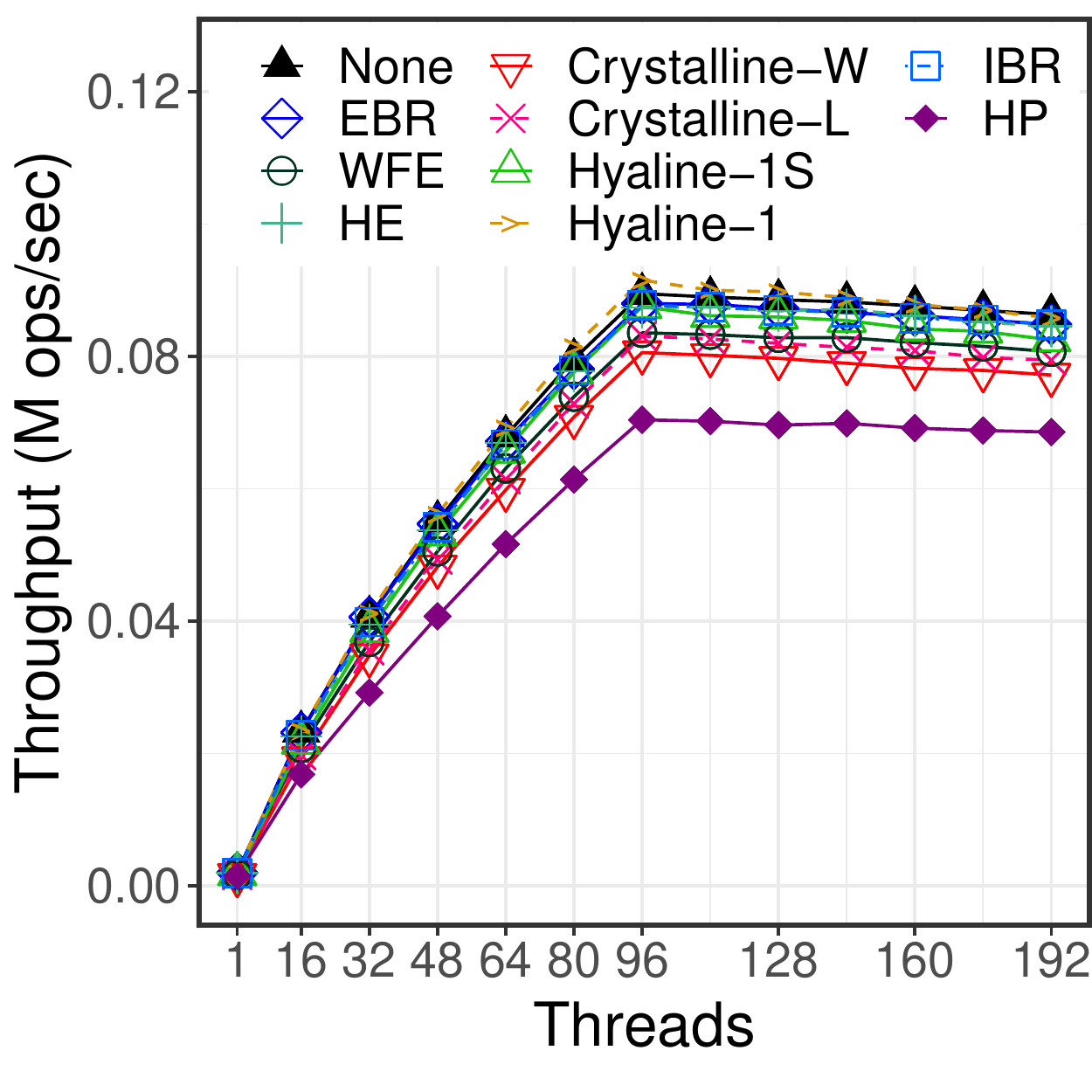}
\vspace{-15pt}
\caption{Throughput (read)}
\label{fig:list_throughput_read}
\end{subfigure}%
\begin{subfigure}{.25\textwidth}
\includegraphics[width=.99\textwidth]{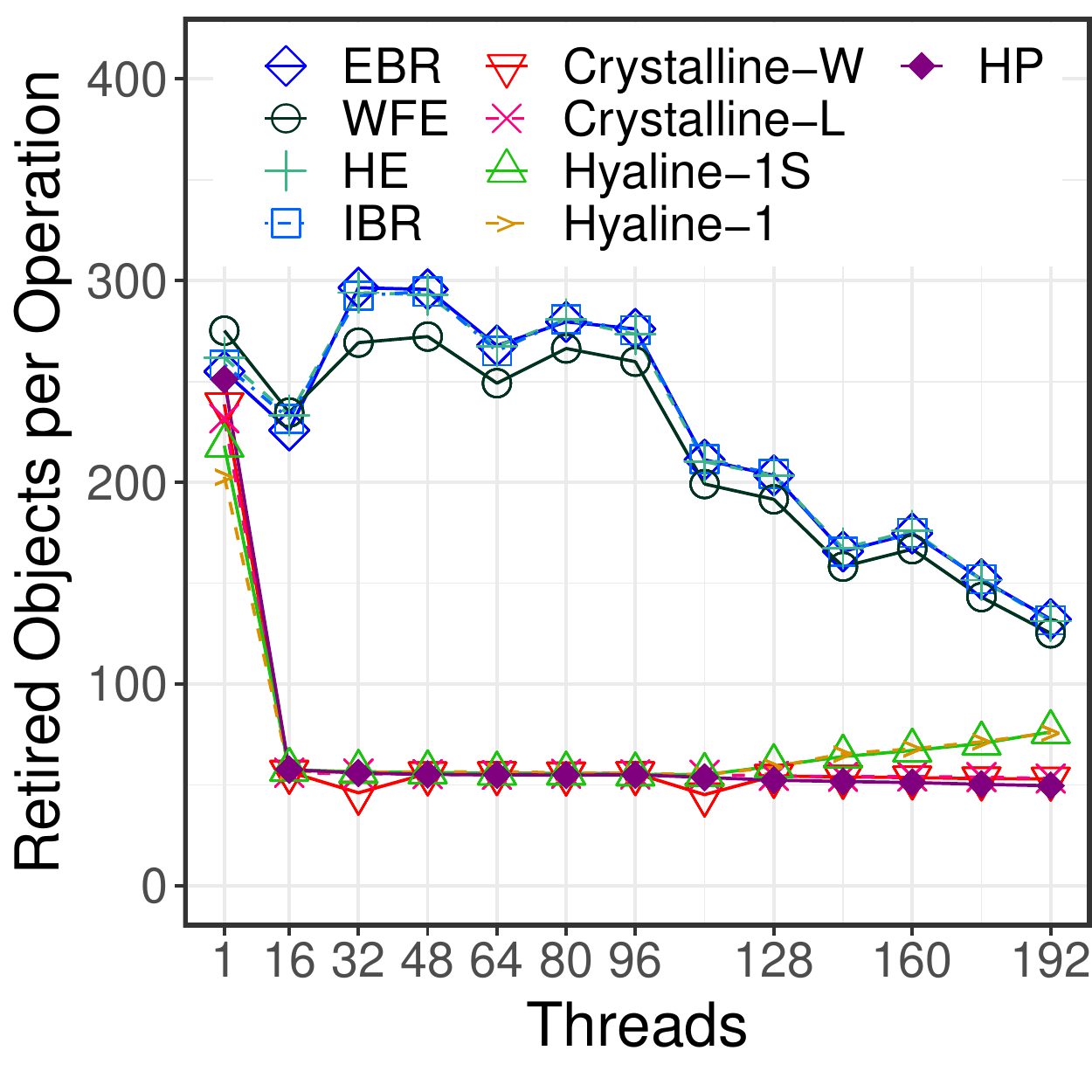}
\vspace{-15pt}
\caption{Retired objs (read)}
\label{fig:list_retired_read}
\end{subfigure}%
\vspace{-5pt}
\caption{Lock-free LinkedList.}
\label{fig:list}
\end{figure*}

\begin{figure*}[ht]
\begin{subfigure}{.25\textwidth}
\includegraphics[width=.99\textwidth]{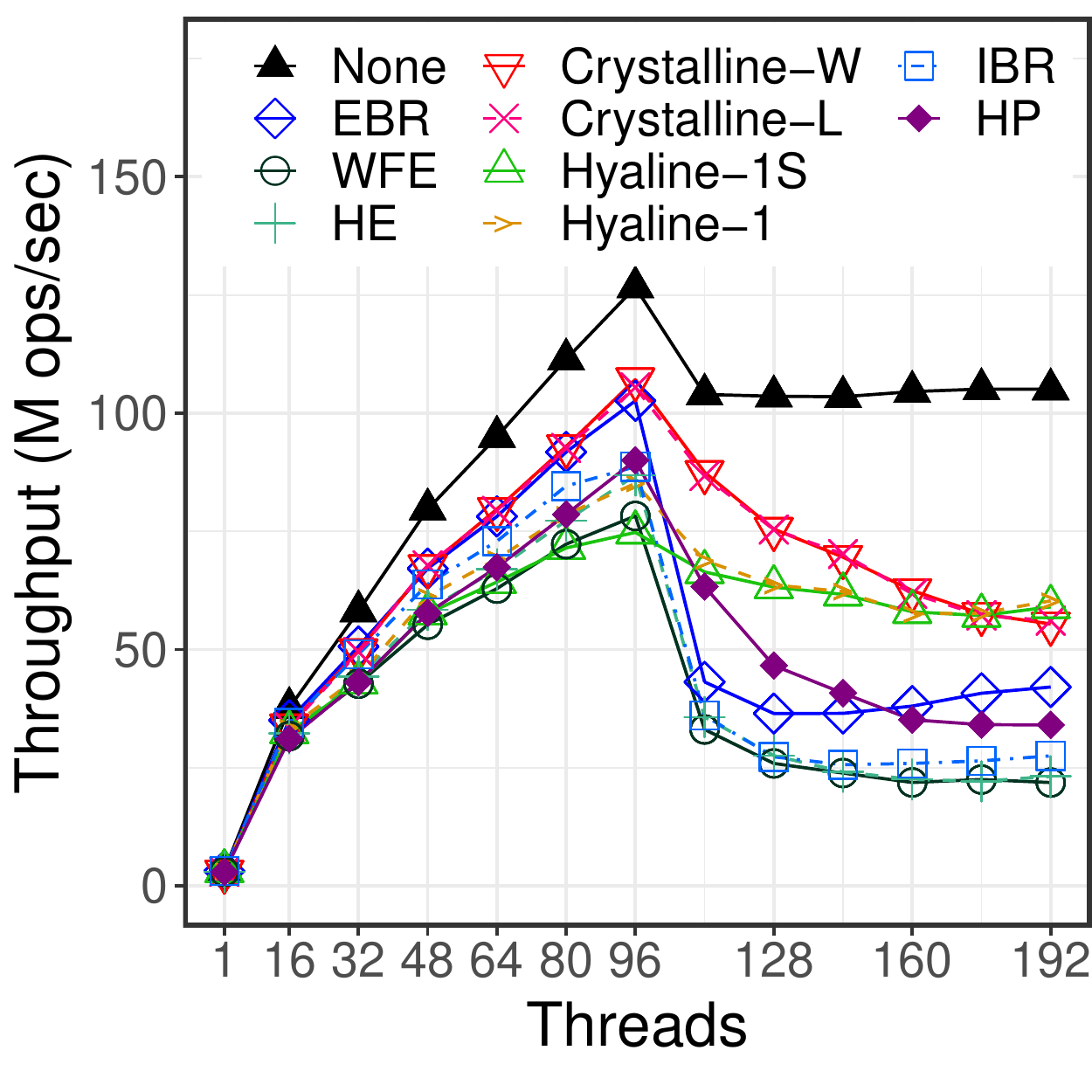}
\vspace{-15pt}
\caption{Throughput (write)}
\label{fig:hashmap_throughput}
\end{subfigure}%
\begin{subfigure}{.25\textwidth}
\includegraphics[width=.99\textwidth]{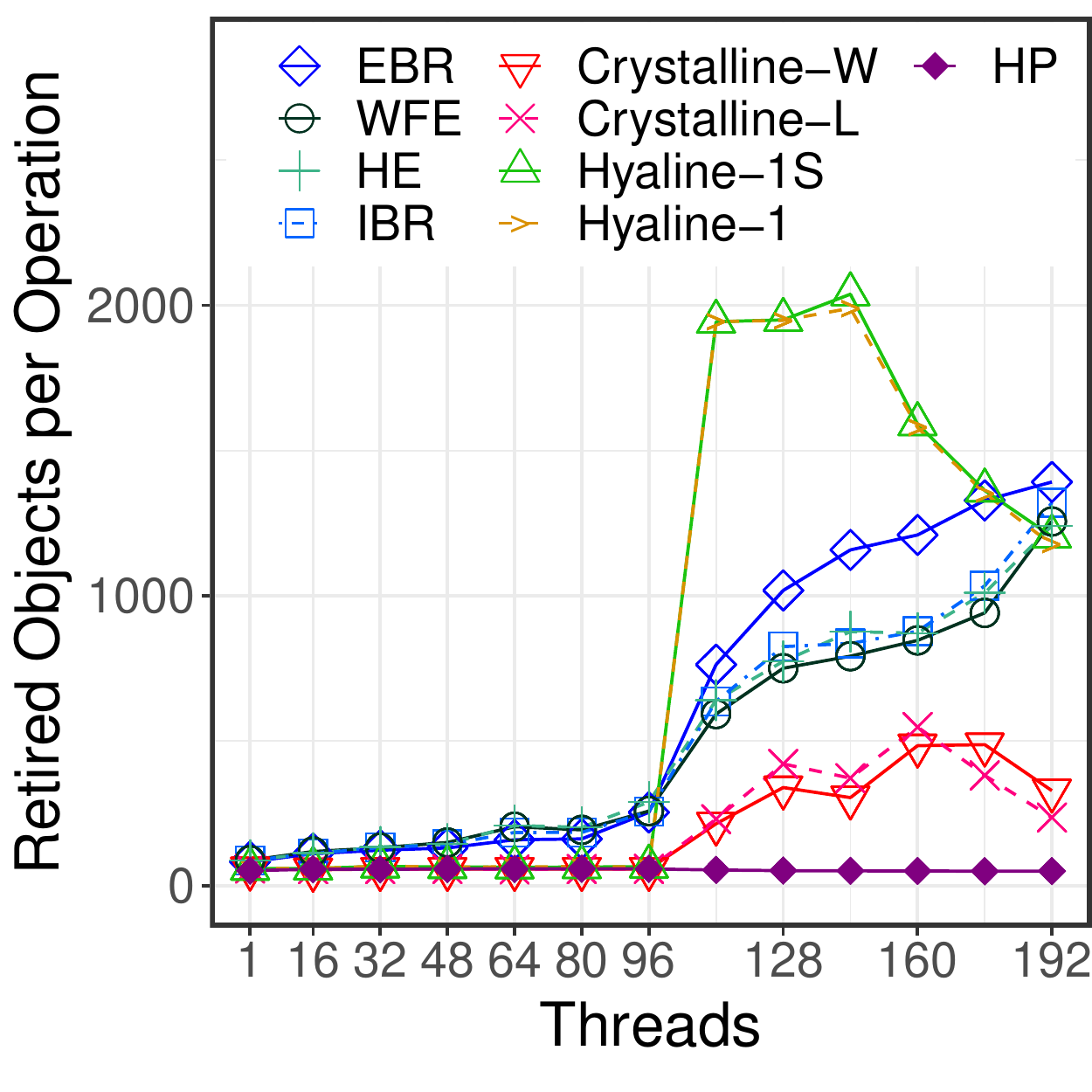}
\vspace{-15pt}
\caption{Retired objs (write)}
\label{fig:hashmap_retired}
\end{subfigure}%
\begin{subfigure}{.25\textwidth}
\includegraphics[width=.99\textwidth]{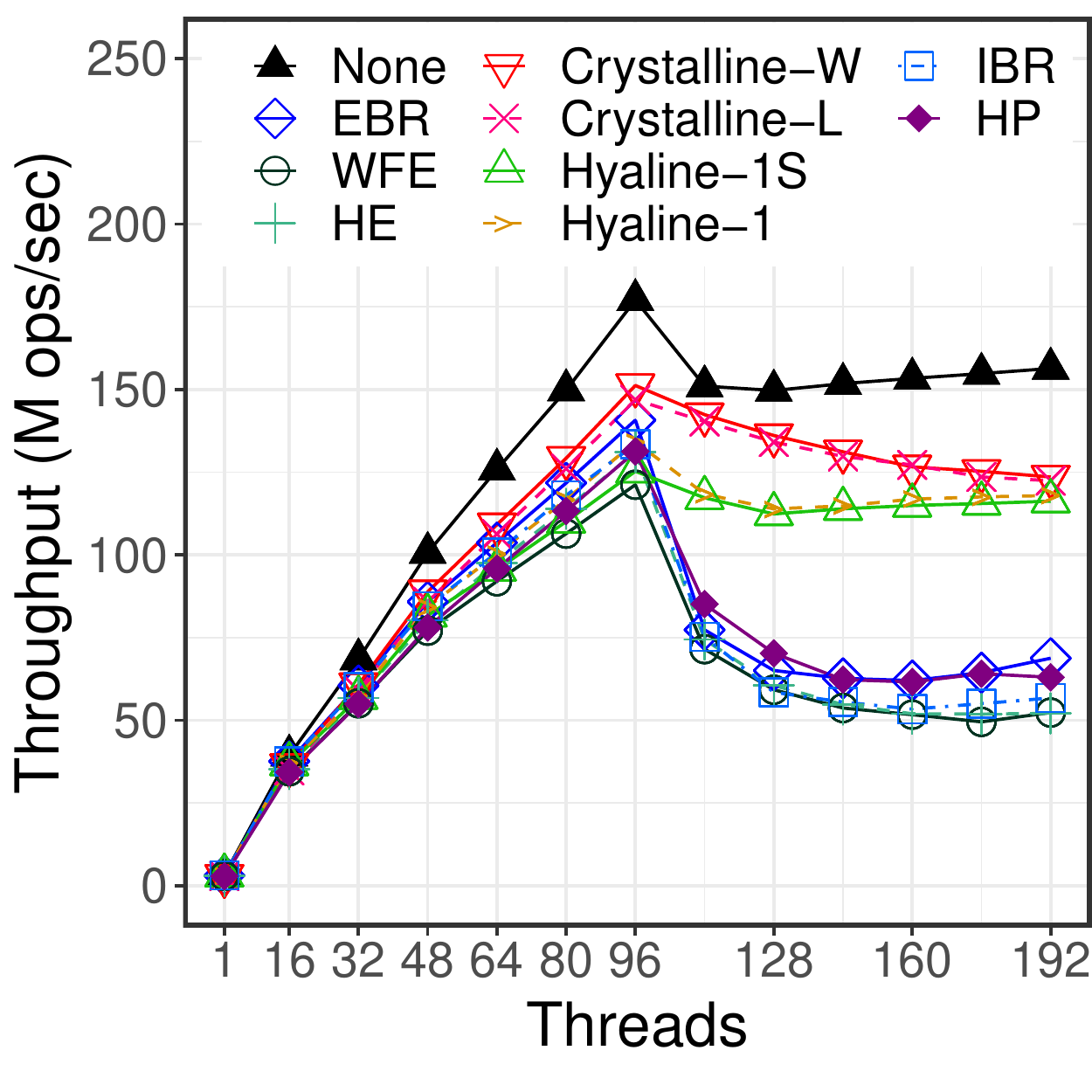}
\vspace{-15pt}
\caption{Throughput (read)}
\label{fig:hashmap_throughput_read}
\end{subfigure}%
\begin{subfigure}{.25\textwidth}
\includegraphics[width=.99\textwidth]{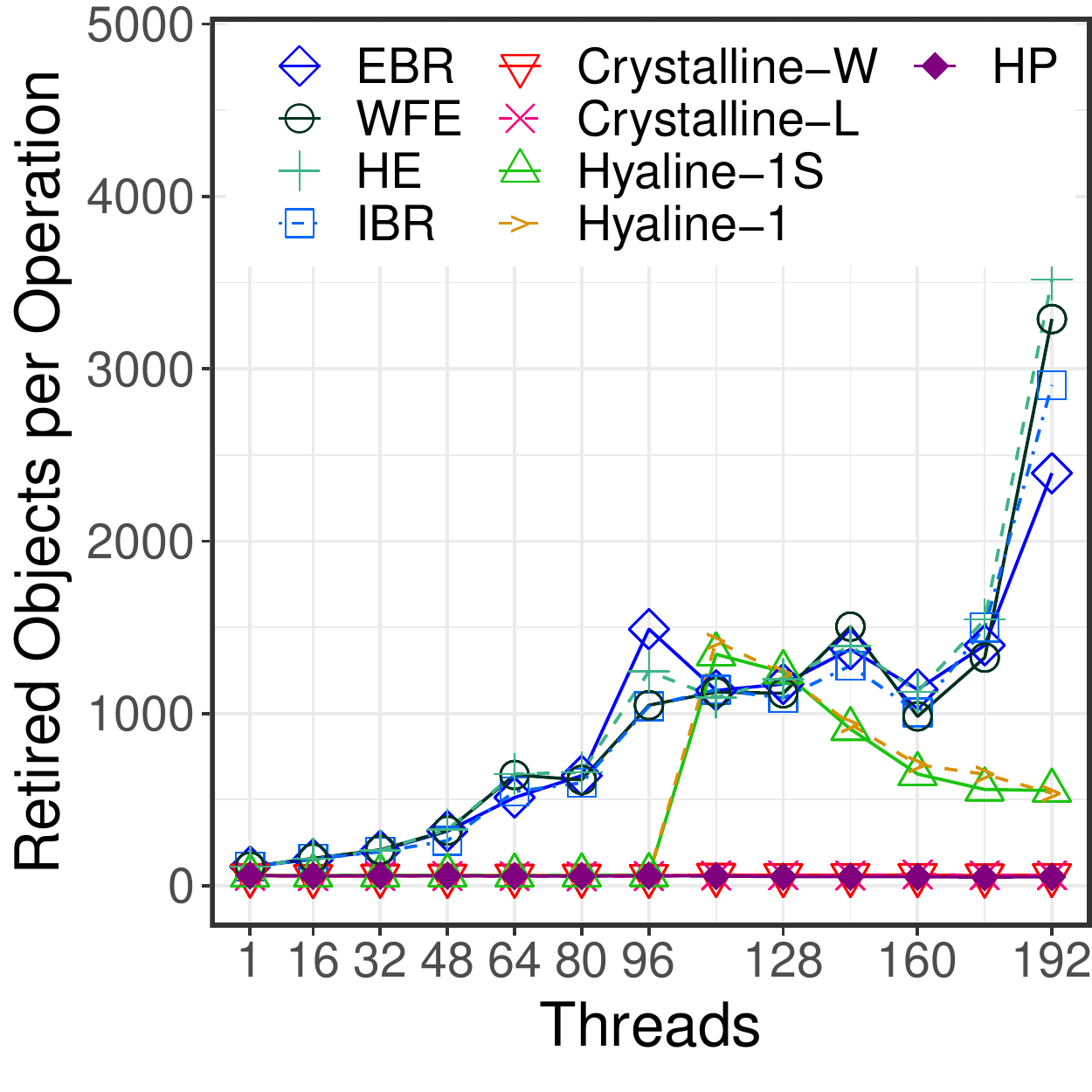}
\vspace{-15pt}
\caption{Retired objs (read)}
\label{fig:hashmap_retired_read}
\end{subfigure}%
\vspace{-5pt}
\caption{Lock-free HashMap.}
\label{fig:hashmap}
\end{figure*}

\begin{figure*}[ht]
\begin{subfigure}{.25\textwidth}
\includegraphics[width=.99\textwidth]{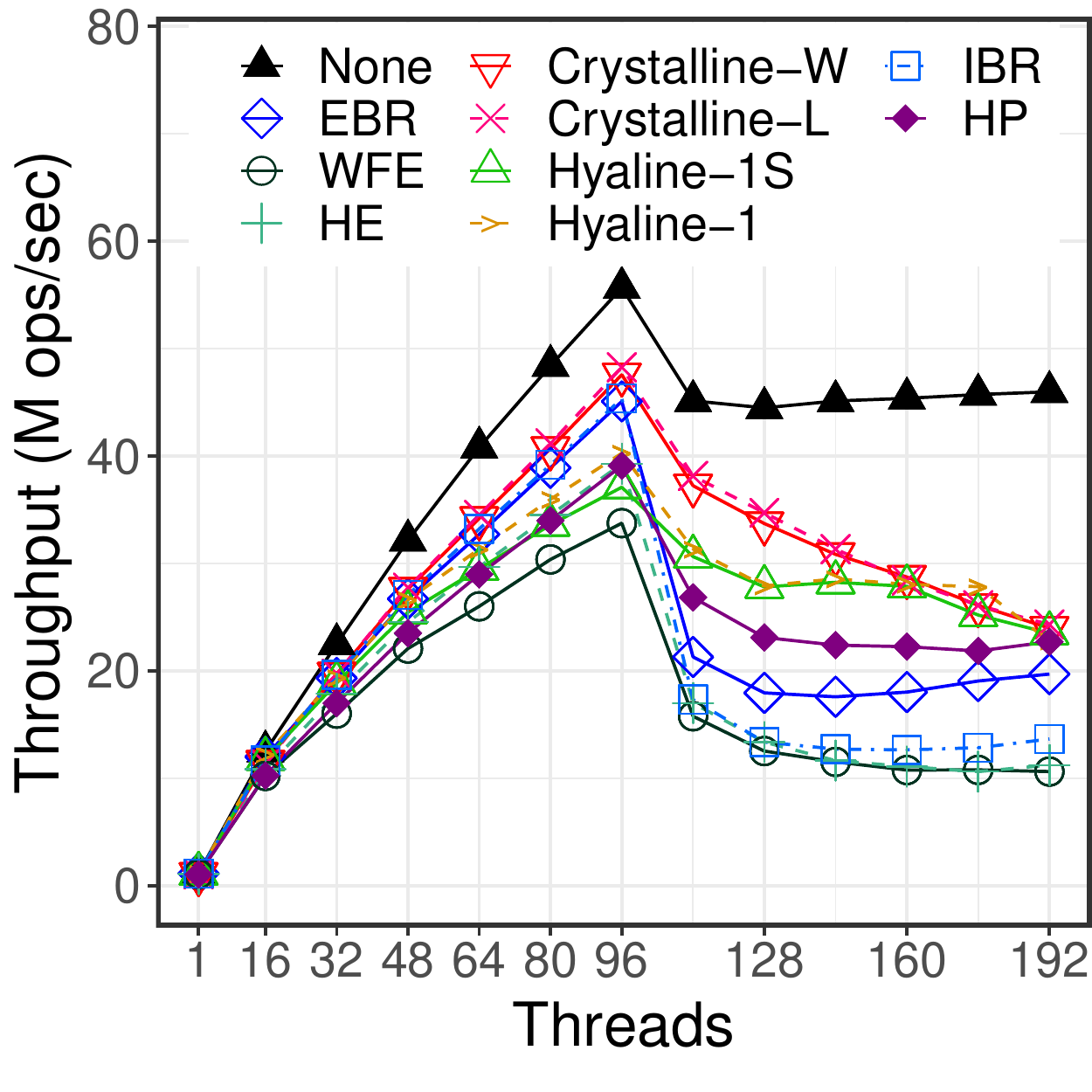}
\vspace{-15pt}
\caption{Throughput (write)}
\label{fig:natarajan_throughput}
\end{subfigure}%
\begin{subfigure}{.25\textwidth}
\includegraphics[width=.99\textwidth]{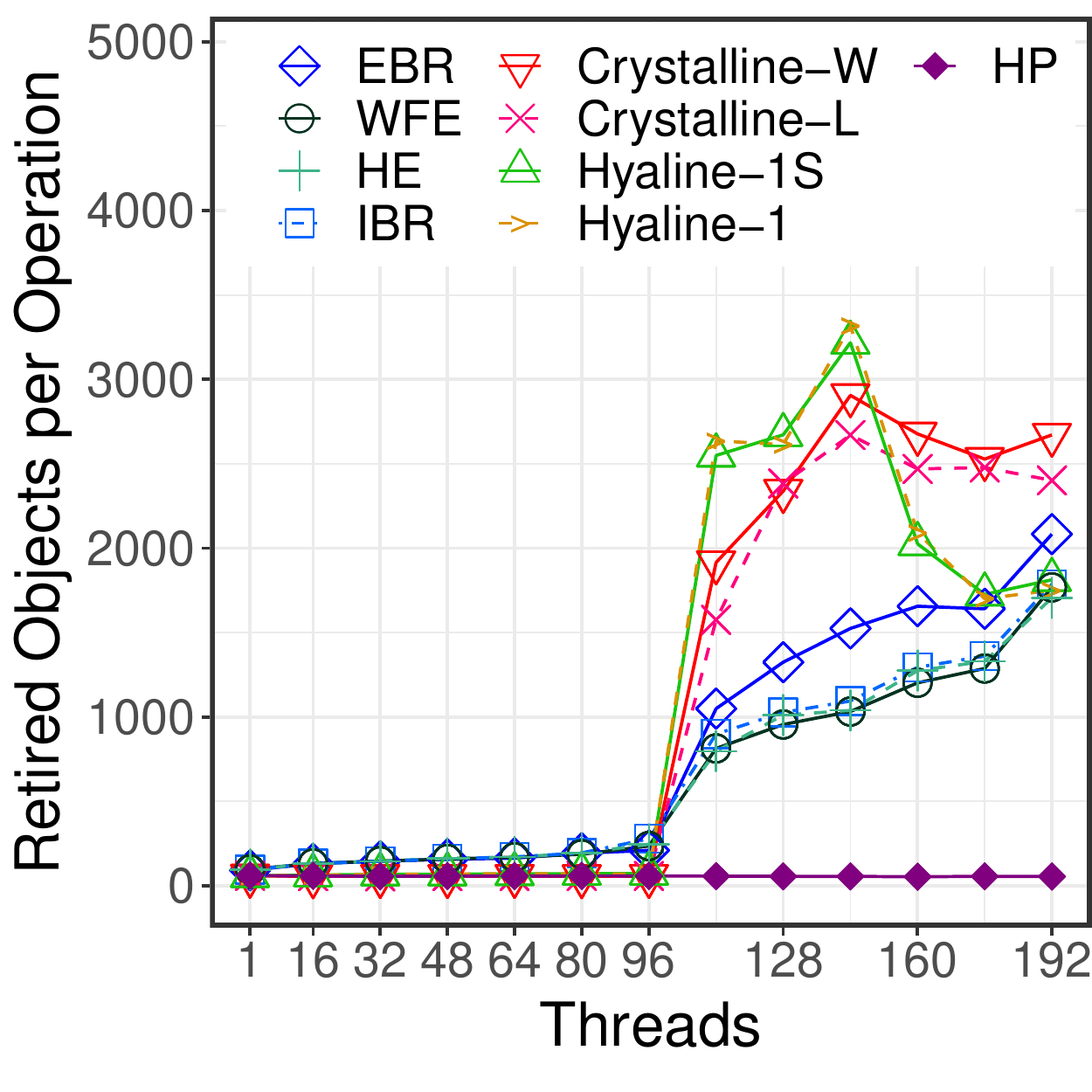}
\vspace{-15pt}
\caption{Retired objs (write)}
\label{fig:natarajan_retired}
\end{subfigure}%
\begin{subfigure}{.25\textwidth}
\includegraphics[width=.99\textwidth]{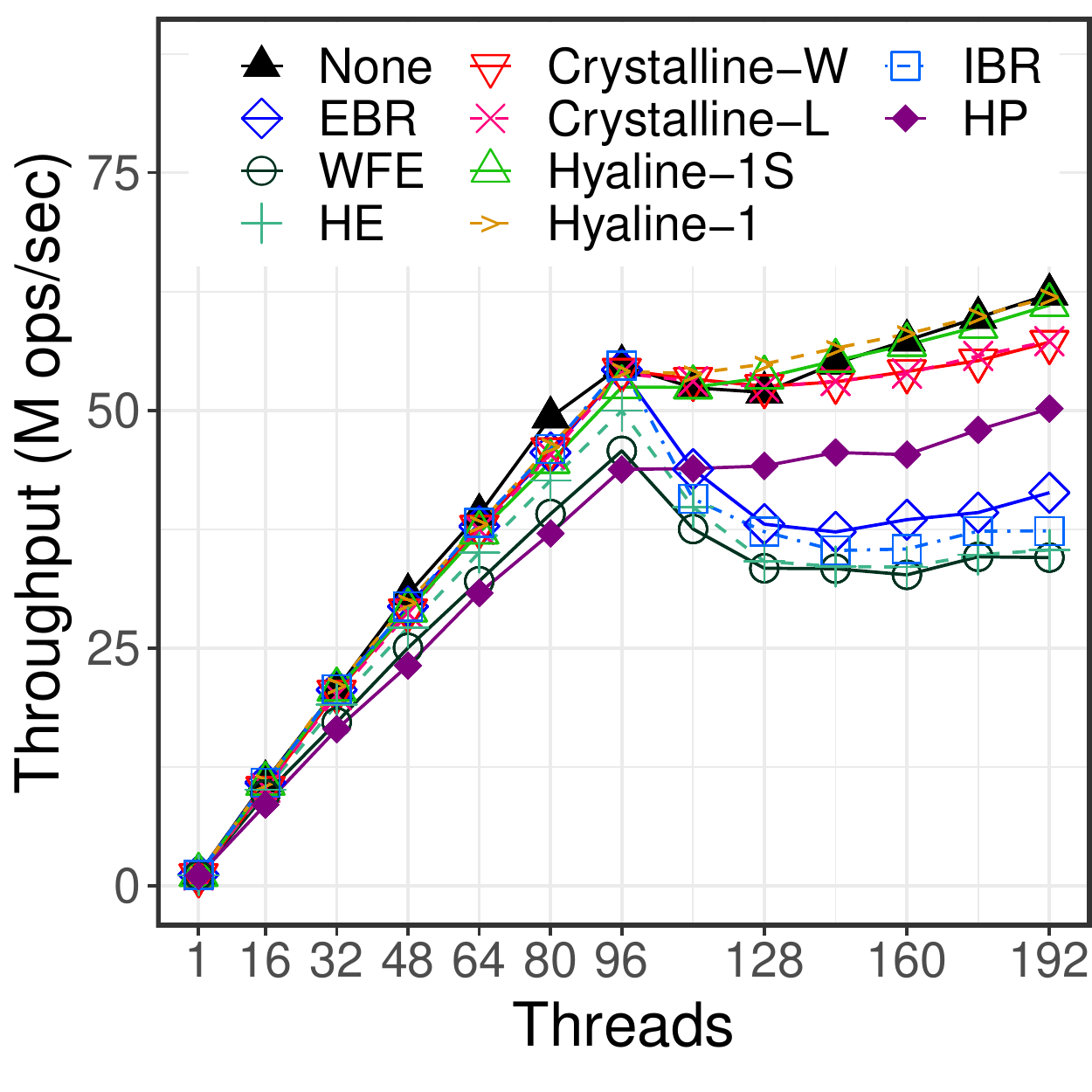}
\vspace{-15pt}
\caption{Throughput (read)}
\label{fig:natarajan_throughput_read}
\end{subfigure}%
\begin{subfigure}{.25\textwidth}
\includegraphics[width=.99\textwidth]{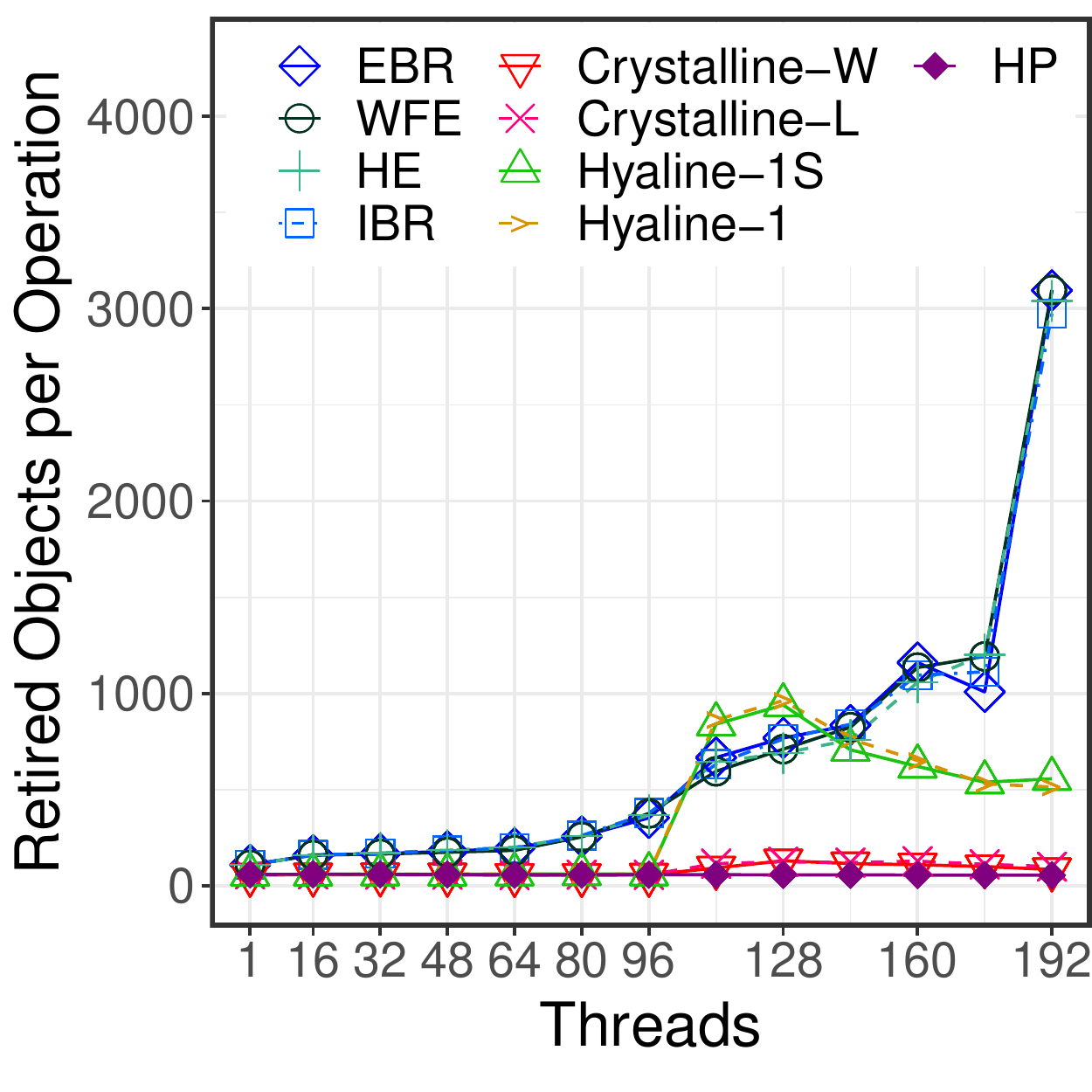}
\vspace{-15pt}
\caption{Retired objs (read)}
\label{fig:natarajan_retired_read}
\end{subfigure}%
\vspace{-5pt}
\caption{Lock-free NatarajanTree.}
\label{fig:natarajan}
\end{figure*}

The benefits of the Crystalline schemes are pronounced even stronger in a read-dominated workload (90\% of \verb|get()| and 90\% of \verb|put()| operations). This is partially due to a better balancing of the reclamation workload across all threads.
As in the write-dominated workload, the hash map (Figures~\ref{fig:hashmap_throughput_read},~\ref{fig:hashmap_retired_read}) and Natarajan
tree (Figures~\ref{fig:natarajan_throughput_read},~\ref{fig:natarajan_retired_read}) achieve the highest throughput, especially in oversubscribed scenarios. At the same time, Crystalline-L/-W achieve exceptional memory efficiency which is on par with hazard pointers. (Even Hyaline-1/1S are visibly less memory efficient.) We only observed a tiny overhead in the linked list shown in Figure~\ref{fig:list_throughput_read}. This overhead is similar to that of WFE~\cite{WFE} and is due to a higher register spillover when \verb|protect()| gets inlined into the code. In linked lists, many nodes need to be traversed, increasing the frequency of \verb|protect()| calls, especially in read-dominated workloads.
However, \verb|protect()| does not update eras that frequently, and a better optimization strategy would be to avoid a premature register spillover. Customized assembly code would likely remove this small overhead altogether.

Overall, the overhead of Crystalline-W (vs. Crystalline-L) is negligible. Both algorithms generally
outperform existing schemes in terms of throughput as well as memory efficiency, which make them appealing for many lock-free and wait-free data structures.
Furthermore, Crystalline-W visibly outperforms WFE, a prior wait-free scheme, even in
non-oversubscribed scenarios. Crystalline-W is also substantially more memory efficient than WFE.

\textbf{Snapshots.}
We disregard snapshot overheads, but for HP/IBR/HE snapshots can create more memory inefficiency than the schemes themselves if the number of threads is high. This makes Crystalline even more efficient than prior schemes.

%% file: related.tex
\section{Related Work}

Existing memory reclamation techniques can be classified into the following
categories.

\textbf{Blocking techniques.}
EBR~\cite{epoch1} is one of the well known approaches, where a thread explicitly makes a
reservation. At the beginning of an operation, EBR records the global epoch value.
At the end, EBR resets the reservation. In quiescent-based reclamation (QSBR)~\cite{epoch2},
this is done automatically as threads go through a ``quiescent'' state.
Modifications to EBR also exist~\cite{Poter:2018:BAS:3210377.3210661} which
bound reclamation costs.
These techniques do not bound memory usage, and thus are blocking when memory is exhausted. Hyaline-1~\cite{Hyaline,hyalineFULL} implements a similar API but can be more memory efficient and faster.

\textbf{Robust techniques.}
To bound memory usage and improve usability, several techniques were developed that exploit OS support. As was pointed out in~\cite{Cohen:2018:DSD:3288538.3276513}, these approaches are not strictly non-blocking, because typical OS primitives such as \emph{signals} use locks internally (e.g., in Linux). ThreadScan~\cite{ThreadScan} is one such mechanism which uses signals. Once the signal is received, every thread scans its own stack and registers to report its working-set to the reclaiming thread. The reclaiming thread uses the collected information to determine unreachable objects that can be reclaimed. Forkscan~\cite{ForkScan} is an extension to ThreadScan that reduces the interruption time to only working threads. Instead of having each thread scan its own stack and registers, threads in Forkscan only write their stack boundaries and register contents, and wait for the forked scanner to begin before resuming work. The forked scanner does the lengthy scanning operation independently. DEBRA+~\cite{DEBRAPaper} and NBR~\cite{NBR} are other examples that use signals.
QSense~\cite{Balmau:2016:FRM:2935764.2935790} relies on the OS scheduler behavior. (Thus, it is hard to guarantee non-blocking behavior in general.) It mixes QSBR~\cite{epoch2} with hazard pointers~\cite{HPPaper}.

IBR~\cite{IBRPaper} and PEBR~\cite{PEBR} improve upon EBR and are not dependent on OS environments. IBR only defends against threads that are stalled indefinitely. Starving threads may still reserve
an unbounded number of blocks. Thus, IBR's authors advise to restart operations
when they fail to make progress. Although restarting is trivial for simple
data structures such as linked lists, it is more problematic for complex
data structures. In the same vein, PEBR's authors demonstrated their scheme
while assuming restarting. PEBR inherently requires restarting to retain simple
EBR-like semantics~\cite{PEBRCOMM}, a central claim of their paper. Moreover,
PEBR's API does not put any explicit bound on how many blocks each thread can
reserve. PEBR's authors only compare against EBR, and PEBR's performance
appears to be only 85-90\% of EBR's, i.e., worse than IBR's and only marginally better than hazard pointers'. Due to potentially unbounded memory usage and restarting, IBR and PEBR are not lock-free in general. Hyaline-1S~\cite{Hyaline,hyalineFULL} provides IBR-like API and has similar progress guarantees but is often faster and more efficient.

\textbf{Lock-free techniques.}
A number of lock-free approaches that bound memory usage were proposed over the years. Traditional reference counting~\cite{refcount4,refcount1,refcount3,refcount2} is fine-grained but has high overheads, especially in read-dominated workloads. Hazard pointers~\cite{HPPaper} and pass-the-buck~\cite{Herlihy:2005:NMM:1062247.1062249,Herlihy:2002:ROP:645959.676129} are also very precise as they track each object individually. However, these techniques still have
high overheads due to their extensive use of memory barriers for each pointer retrieval.
The original paper on hazard pointers~\cite{HPPaper} advertises hazard pointers as a ``wait-free'' scheme, but
pointer retrievals are only lock-free in general, unless the data structure itself takes
steps to guarantee wait-freedom as in~\cite{pedroWFQUEUE}.
Some approaches~\cite{Braginsky:2013:DAL:2486159.2486184} aim to reduce overheads
of hazard pointers, but they are only suitable for specific data structures such as linked lists.
Other techniques~\cite{Cohen:2015:AMR:2814270.2814298,Cohen:2015:EMM:2755573.2755579} do not have this limitation, but require data structures to be represented in a normalized form~\cite{10.1145/2555243.2555261}. This, however, can be burdensome.
FreeAccess~\cite{Cohen:2018:DSD:3288538.3276513} removes this burden and uses a garbage collector. However, it does not transparently handle SWAP, which coincidentally is a prerequisite for making our Crystalline-W algorithm wait-free. OrcGC~\cite{OrcGC} is another lock-free garbage collector with great performance, but it can be slower in some tests than hazard pointers.
Hazard eras~\cite{HEPaper} is a lock-free scheme with bounded memory usage, which is inspired by hazard pointers but uses epochs to expedite the algorithm. Our Crystalline-L scheme is compatible with hazard eras/pointers, provides similar progress guarantees but is often faster and more efficient.

\textbf{Wait-free techniques.}
None of the above schemes are wait-free, which makes it difficult to use them in wait-free data structures. Wait-free reclamation has recently received increasing attention. OneFile~\cite{OneFile} implements Software Transactional Memory (STM) with wait-free reclamation. CX~\cite{CX} implements a universal construct which converts the sequential specification of a data structure into a wait-free implementation. Both OneFile and CX implement their own, specialized wait-free reclamation schemes. Although these approaches enable the implementation of many wait-free data structures, customized, hand-crafted data structures can often better utilize parallelism and
achieve higher overall performance. To that end, a wait-free reclamation approach, WFE~\cite{WFE} was recently proposed. Crystalline-W is partially inspired by WFE's design, but goes beyond WFE by achieving high memory efficiency and performance over a broader range of conditions as we have shown in this paper.

%% file: conclusion.tex
\section{Conclusions}

We presented a new wait-free memory reclamation scheme, Crystalline-W, which
guarantees complete wait-freedom with bounded memory usage. Crystalline-W's uniquely distinguishing aspect is that it incorporates all desirable properties of prior schemes including wait-freedom, asynchronous reclamation, and balanced reclamation workload in the \textit{same} algorithm. The only existing wait-free scheme, WFE, lacks the two latter properties. Unsurprisingly, Crystalline-W outperforms WFE in almost all test cases.

Crystalline-W is based on a lock-free algorithm, Crystalline-L, which is also presented in the paper. Crystalline-L is an improved version of Hyaline-1S that additionally guarantees bounded memory usage even in the presence of starving threads. Furthermore, Crystalline-L introduces \textit{dynamic batches} which resolve one major inconvenience with Hyaline-1S and also improve overall performance.

Both Crystalline-L and Crystalline-W demonstrate very high throughput and great memory efficiency, unparalleled among existing schemes, which is especially evident in read-dominated workloads. Crystalline-L's performance is even superior to that of Hyaline-1S, which proves the benefits of dynamic batches and a more fine-grained API.

\section*{Availability}
Crystalline's source code is available at \url{https://github.com/rusnikola/wfsmr}.

\section*{Acknowledgments}
This work is supported in part by ONR under grants N00014-18-1-2022 and N00014-19-1-2493, and AFOSR under grant FA9550-16-1-0371.